\documentclass[11pt]{article}

\usepackage[letterpaper, margin=1in]{geometry}
\usepackage{amsfonts,amssymb,amsmath,amsthm}
\usepackage[utf8]{inputenc}
\usepackage[T1]{fontenc}
\usepackage{newtxtext}
\usepackage[hidelinks]{hyperref}
\usepackage[shortlabels]{enumitem}
\usepackage{color}
\usepackage{mathtools}
\usepackage{thm-restate}
\usepackage{lmodern}
\usepackage[nameinlink,capitalise]{cleveref}
\theoremstyle{plain}
\newtheorem{theorem}{Theorem}[section]
\newtheorem{lemma}[theorem]{Lemma}
\newtheorem{definition}[theorem]{Definition}
\newtheorem*{definition*}{Definition}
\newtheorem{corollary}[theorem]{Corollary}

\newtheorem{claim}[theorem]{Claim}
\theoremstyle{remark}
\newtheorem*{rem}{Remark}

\newtheorem{exmp}{Example}
\newtheorem*{exmp*}{Example}

\newcommand{\N}{\mathbb{N}}
\newcommand{\Z}{\mathbb{Z}}
\newcommand{\Q}{\mathbb{Q}}

\newcommand{\C}{\mathbb{C}}

\newcommand{\F}{\mathbb{F}}

\newcommand{\poly}{\mathrm{poly}}
\DeclareMathOperator{\codim}{\mathrm{codim}}

\newcommand{\aff}{\mathbb{A}}
\newcommand{\proj}{\mathbb{P}}
\newcommand{\gr}{\mathrm{G}}
\newcommand{\spn}{\mathrm{span}}
\newcommand{\pgr}{\mathbb{G}}

\newcommand{\simple}{\mathrm{sim}}

\newcommand{\eps}{\varepsilon}
\renewcommand{\epsilon}{\eps}

\begin{document}
 
\title{Variety Evasive Subspace Families\footnote{A previous version of this paper \cite{Guo21} appeared in the 36th Computational Complexity Conference (CCC 2021).}}

\author{Zeyu Guo\footnote{Department of Computer Science and Engineering, The Ohio State University.
Email: \href{mailto:zguotcs@gmail.com}{\texttt{zguotcs@gmail.com}}}
}

\date{}

\maketitle

\begin{abstract}
We introduce the problem of constructing explicit \emph{variety evasive subspace families}. 
Given a family $\mathcal{F}$ of subvarieties of a projective or affine space, a collection $\mathcal{H}$ of projective or affine $k$-subspaces is \emph{$(\mathcal{F},\epsilon)$-evasive} if
for every $\mathcal{V}\in\mathcal{F}$, all but at most $\epsilon$-fraction of $W\in\mathcal{H}$ intersect every irreducible component of $\mathcal{V}$ with (at most) the expected dimension.
The problem of constructing such an explicit subspace family generalizes both deterministic black-box polynomial identity testing (PIT) and the problem of constructing explicit (weak) lossless rank condensers. 

Using Chow forms, we construct explicit  $k$-subspace families of polynomial size that are evasive for all varieties of bounded degree in a projective or affine $n$-space. As one application, we obtain a complete derandomization of Noether's normalization lemma for varieties of low degree in a projective or affine $n$-space. In another application, we obtain a simple polynomial-time black-box PIT algorithm for depth-$4$ arithmetic circuits with bounded top fan-in and bottom fan-in that are not in the Sylvester--Gallai configuration, improving and simplifying a result of Gupta (ECCC TR 14-130).

As a complement of our explicit construction, we prove a tight lower bound for the size of $k$-subspace families that are evasive for degree-$d$ varieties in a projective $n$-space. When $n-k=n^{\Omega(1)}$, the lower bound is superpolynomial unless $d$ is bounded. 
The proof uses a dimension-counting argument on Chow varieties that parametrize projective subvarieties.
\end{abstract}

\section{Introduction}

Polynomial identity testing  (PIT) is a fundamental problem in the areas of derandomization and algebraic complexity theory. 
The problem asks whether a multivariate polynomial, computed by an arithmetic circuit, formula, or other algebraic computational models, is identically zero.  For example, the polynomial $(X+Y)(X-Y) - X^2-Y^2$ is identically zero while $(X+Y)^2-X^2$ is not. 
%
%For example, the polynomial $(x+y)(x-y) - x^2-y^2$ is identically zero. 
%
%The input to the problem can be given as an algebraic formula or circuit or 
%other algebraic computation
%models like arithmetic branching programs or determinant of a symbolic matrix.
%

It is easy to solve PIT in randomized polynomial time, as we may simply evaluate  the input polynomial at a random point and check if the evaluation is zero.
On the other hand, finding a deterministic polynomial-time PIT algorithm for general arithmetic circuits is  a long-standing open problem.
Such algorithms are known for some special cases, and we refer the readers to the surveys \cite{Sax09, Sax13, SY10} for  details.

\emph{Black-box} PIT algorithms are a special kind of PIT algorithm.
A (deterministic) black-box PIT algorithm tests if a polynomial in a family $\mathcal{F}$ is zero by constructing a \emph{hitting set}  for $\mathcal{F}$, which is a finite collection $\mathcal{H}$ of evaluation points with the following property: for any nonzero  $Q\in \mathcal{F}$, there exists $p\in\mathcal{H}$ such that the evaluation of $Q$ at $p$ is nonzero. 
After constructing a hitting set $\mathcal{H}$ for $\mathcal{F}$, the algorithm simply checks if the evaluation of the given polynomial at every point in $\mathcal{H}$ is zero. 
The problem of designing a deterministic black-box PIT algorithm for polynomials in $\mathcal{F}$ is thus equivalent to deterministically constructing a hitting set for $\mathcal{F}$.
To make the algorithm efficient, such a hitting set should be small and efficiently computable. 
Hitting sets are also called \emph{correct test sequences} and are studied in \cite{HS80, PS22}.

%It is often useful for the hitting set $\mathcal{H}$ to have the following stronger property: for any nonzero $Q\in \mathcal{F}$,  the evaluation $Q(p)$ is nonzero for all but at most $\epsilon$-fraction of $p\in\mathcal{H}$. We say a finite collection $\mathcal{H}$ of evaluation points is an \emph{$\epsilon$-hitting set} for $\mathcal{F}$ if it has this stronger property.\footnote{This terminology is not standard.} Many explicit constructions of hitting sets actually yield $\epsilon$-hitting sets.

From a geometric perspective,  an $n$-variate nonzero polynomial $Q$ over an algebraically-closed field $\F$ defines  a \emph{hypersurface} $\mathcal{V}(Q):=\{\alpha\in \F^n: Q(\alpha)=0\}$ of $\F^n$.
A hitting set $\mathcal{H}$ for $\mathcal{F}$ has the property that for every nonzero $Q\in\mathcal{F}$, there exists a point $p\in \mathcal{H}$ that is disjoint from the hypersurface $\mathcal{V}(Q)$, or we say $p$ \emph{evades} $\mathcal{V}(Q)$.  It is natural to consider the generalization of this property to higher dimensions/codimensions. Namely, we want to construct a finite collection $\mathcal{H}$ of  \emph{affine $k$-subspaces} (i.e. affine subspaces of dimension $k$) such that for every  \emph{variety} $\mathcal{V}\subseteq\F^n$ (i.e., solution set of a set of polynomial equations) from a certain family, some (or most) $W\in \mathcal{H}$ \emph{evade} $\mathcal{V}$, in the sense that the dimension of the intersection $\mathcal{V}\cap W$ is bounded by the expected dimension achieved by $W$ in general position.
A similar property can be defined for projective $k$-spaces, to be defined below.
We call such a collection $\mathcal{H}$ of projective or affine $k$-subspaces a \emph{variety evasive subspace family}. The formal definition is given below.

%for every variety $\mathcal{V}$ of codimension $k+1$ from a certain family $\mathcal{F}$, most $W\in\mathcal{H}$ are disjoint from $\mathcal{V}$. More generally, we could drop the assumption $\codim(\mathcal{V})=k+1$ instead require $\dim(\mathcal{V}\cap W)\leq k-\codim(\mathcal{V})$ for most $W\in\mathcal{H}$.

\subsection{Variety Evasive Subspace Families}

Let $\F$ be an algebraically closed field.
An \emph{affine $n$-space} $\aff^n$, as a set, is simply defined to be the vector space $\F^n$.
We also need the notion of a \emph{projective $n$-space}, denoted by $\proj^n$, which is (intuitively) the set of lines passing through the origin $\mathbf{0}$ of $\aff^{n+1}$. Formally, it is defined to be  the quotient set $(\aff^{n+1}\setminus\{\mathbf{0}\})/\sim$, where $\sim$ is the equivalence relation defined by scaling, i.e., $u\sim v$ if $u=cv$ for some nonzero scalar $c\in\F$.

An \emph{(affine) subvariety} $\mathcal{V}\subseteq\aff^n$ is   the set of  common zeros of a set of $n$-variate polynomials over $\F$. 
Similarly, a \emph{(projective) subvariety} $\mathcal{V}\subseteq\proj^n$ is  the set of common zeros  of a set of \emph{homogeneous} $(n+1)$-variate polynomials  over $\F$, where we represent each element of $\proj^n$ as an $(n+1)$-tuple in $\aff^{n+1}$. 
In this paper, a \emph{variety} refers to a  subvariety of a projective or affine space, and is said to be \emph{irreducible} if it cannot be written as a union of finitely many proper subvarieties.\footnote{Varieties in this paper  are   not necessarily irreducible. They are often called \emph{algebraic sets} in literature.}
%Denote by $\proj^n$ and $\aff^n$ the projective $n$-space and the affine $n$-space over $\F$ respectively.

The \emph{dimension}  of a variety $\mathcal{V}$, denoted by $\dim(\mathcal{V})$, is intuitively the ``degree of freedom'' of picking a point in the variety.
See \cref{sec_ag} for its formal definition.
In particular, for a linear subspace $V\subseteq\aff^n$, the dimension of $V$ as a variety is simply its linear-algebraic dimension.

For  two  irreducible subvarieties $\mathcal{V}_1$ and $\mathcal{V}_2$ of $\proj^n$ or $\aff^n$ in \emph{general position}, we  expect  the dimension of $\mathcal{V}_1\cap\mathcal{V}_2$ to be $ \dim(\mathcal{V}_1) + \dim(\mathcal{V}_2) - n$ (unless $ \dim(\mathcal{V}_1) + \dim(\mathcal{V}_2) < n$, in which case we expect $\mathcal{V}_1\cap\mathcal{V}_2=\emptyset$).
The following definition captures the condition that $\dim(\mathcal{V}_1\cap \mathcal{V}_2)$ is bounded by the expected dimension.

\begin{definition}[Evading]\label{defi_evade}
Let $\mathcal{V}_1$ and $\mathcal{V}_2$ be irreducible subvarieties of $\proj^n$ or $\aff^n$.
We say $\mathcal{V}_1$ \emph{evades} $\mathcal{V}_2$ if 
\[
\dim(\mathcal{V}_1\cap\mathcal{V}_2)\leq \dim(\mathcal{V}_1) + \dim(\mathcal{V}_2) - n,
\]
where the dimension of an empty set is assumed to be $-\infty$.
In particular, if $\dim(\mathcal{V}_1) + \dim(\mathcal{V}_2) < n$, then $\mathcal{V}_1$ evades $\mathcal{V}_2$ iff $\mathcal{V}_1\cap \mathcal{V}_2=\emptyset$.

More generally, suppose $\mathcal{V}_1$ is irreducible but $\mathcal{V}_2$ is possibly reducible. We say $\mathcal{V}_1$ \emph{evades} $\mathcal{V}_2$ if it evades every irreducible component of $\mathcal{V}_2$. 
\end{definition}

Next, we define \emph{subspace families} and \emph{variety evasive subspace families}.

\begin{definition}[Subspace family]
For $0\leq k\leq n$, a finite collection\footnote{In this paper, a \emph{collection} is a multiset, i.e., its elements are allowed to appear more than once. 
%The cardinality of a finite collection is the sum of the multiplicities of its elements.
} 
of  $k$-subspaces of $\proj^n$ is called a \emph{(projective) $k$-subspace family} on $\proj^n$. Similarly,  a finite collection of affine $k$-subspaces of $\aff^n$ is called an \emph{affine $k$-subspace family} on $\aff^n$.
\end{definition}

%\TODO{Change ``set'' to ``collection''}

\begin{definition}[Variety evasive subspace family]\label{defi_evasive}
Let $\mathcal{F}$ be a family of subvarieties of $\proj^n$ (resp. $\aff^n$).
Let $\mathcal{H}$ be a $k$-subspace family on $\proj^n$ (resp. affine $k$-subspace family on $\aff^n$) where $0\leq k\leq n$.
Then:
\begin{itemize}
\item We say $\mathcal{H}$ is  \emph{$\mathcal{F}$-evasive}  if for every $\mathcal{V}\in\mathcal{F}$, 
there exists $W\in\mathcal{H}$ that evades $\mathcal{V}$. 
\item  We say  $\mathcal{H}$ is  \emph{$(\mathcal{F}, \epsilon)$-evasive}  if for every $\mathcal{V}\in\mathcal{F}$, 
a random element $W\in\mathcal{H}$ evades $\mathcal{V}$ with probability at least $1-\epsilon$.
\end{itemize}
\end{definition}

\paragraph{Connection with hitting sets.}
\cref{defi_evasive} naturally generalizes the notions of hitting sets in the context of PIT. For example, a collection of points in $\proj^n$ is a hitting set for a family $\mathcal{F}$ of homogeneous polynomials in $\F[X_1,\dots,X_{n+1}]$ iff it is an $\mathcal{F}'$-evasive $0$-subspace family, where $\mathcal{F}'=\{\mathcal{V}(P): P\in\mathcal{F}\}$ is the family of hypersurfaces defined by the polynomials in $\mathcal{F}$. In other words, hitting sets may be viewed as $0$-subspace families that are evasive for varieties of codimension one. 

\paragraph{Connection with lossless rank condensers.} Other than the case of codimension one, we may also consider the special case of degree one, and this leads to another important family of pseudorandom objects, called \emph{(weak) lossless rank condensers} \cite{GR08, FS12, FSS14, FG15}. These objects were used by Gabizon and Raz \cite{GR08} to construct affine extractors. They also play a crucial role in polynomial identity testing \cite{KS11, SS12, FS12, FSS14}.

A lossless rank condenser is defined as follows:
Let $r\leq t\leq n$ be positive integers.
A finite collection $\mathcal{H}$ of matrices $E\in \F^{t\times n}$ is called an \emph{$(r, L)$-lossless rank condenser} if for every matrix $M\in \F^{n\times r}$ of rank $r$, the number of $E\in \mathcal{H}$ satisfying $\mathrm{rank}(EM)<r$ is at most $L$.

The connection between lossless rank condensers and variety evasive subspace families can be seen as follows: 
Let us assume every matrix $E\in \mathcal{H}$ has full rank $t$. 
Such a matrix $E$ corresponds a linear $t$-subspace $W$ of $\F^n$.
On the other hand, a matrix $M\subseteq \F^{n\times r}$ of rank $r$ corresponds to a linear $(n-r)$-subspaces of $\F^n$ via $M\mapsto \ker(M)$, where $\ker(M)=\{u\in \F^n: uM=0\}$ denotes the left kernel of $M$.
It is easy to see that the condition $\mathrm{rank}(EM)=r$ is equivalent to $\dim(W\cap \ker(M))=t-r$.
Passing from $\F^n$ to $\proj^{n-1}$ by taking the quotient modulo scalars, this condition is also equivalent to the condition that the two projective subspaces $\proj W$ and $\proj(\ker(M))$ evade each other.

Every projective $(n-r-1)$-subspace of $\proj^{n-1}$ can be realized as  $\proj(\ker(M))$  for some rank-$r$ matrix $M$.
Therefore, $\mathcal{H}$ is an $(r, L)$-lossless rank condenser iff it is an $(\mathcal{F}, \epsilon)$-evasive $(t-1)$-subspace family on $\proj^{n-1}$, where $\epsilon=L/|\mathcal{H}|$ and $\mathcal{F}$ is the family of all $(n-r-1)$-subspaces of $\proj^{n-1}$.

%and $M\in \F^{n\times r}$ of rank $r$, let $W\subseteq \F^n$ be the row space of $E$ and let $V=\{x\in\F^n: xM=0\}$. Then the condition $\mathrm{rank}(EM)=r$ is equivalent to $\dim (W\cap V)=\dim(W)-r$. The latter is equivalent to $\dim(\proj W\cap \proj V)=\dim(\proj W) - r$\footnote{If $\dim(W)=r$}, where $\proj W$ and $\proj V$ are projective subspaces of $\proj^{n-1}$.

Rank condensers are central objects in the theory of ``linear-algebraic pseudorandomness'' coined by Guruswami and Forbes \cite{FG15}. Our study of variety evasive subspace families may be seen as one step of extending the theory to a nonlinear setting.

Explicit lossless rank condensers were used to construct explicit (deterministic) \emph{affine extractors} \cite{GR08} and more generally, \emph{extractors for varieties} \cite{Dvi12}. Similar ideas were used to construct explicit \emph{deterministic extractors} (and \emph{rank extractors}) for \emph{polynomial sources} \cite{DGW09}, which also generalize affine extractors.
It is an interesting question to us whether explicit variety evasive subspace families and the related derandomized Noether's normalization lemma (see below) can be similarly useful in this area.

\subsection{Our Results}\label{sec_result}

We have seen that variety evasive subspace families generalize some important and well-studied pseudorandom objects.
This leads to the following natural question: For which interesting families $\mathcal{F}$ of subvarieties can we construct explicit $\mathcal{F}$-evasive or $(\mathcal{F},\epsilon)$-evasive subspace families?

In this paper, we focus on the families of subvarieties of \emph{bounded degree}.
First, we recall the definition of the \emph{degree} of a variety.
\begin{definition}[degree]
The  degree of an irreducible variety $\mathcal{V}$ in $\proj^n$ (resp. $\aff^n$) is the number of intersections of $\mathcal{V}$ with a general projective (resp. affine) subspace of codimension $\dim(\mathcal{V})$. 
Following \cite{HS80, Hei83}, we define the degree of a (possibly reducible) variety
% $\mathcal{V}$, denoted by $\deg (\mathcal{V})$, 
 to be the sum of the degrees of its irreducible components.
\end{definition}

 For convenience, we  introduce the following definition.

\begin{definition}\label{defi_bounddeg}
We say a projective (resp. affine) $k$-subspace family $\mathcal{H}$ on $\proj^n$ (resp. $\aff^n$) is \emph{$(n,d)$-evasive} if it is $\mathcal{F}$-evasive, where  $\mathcal{F}$ is chosen to be the family of all subvarieties of $\proj^n$ (resp.  $\aff^n$) of   degree at most $d$. 
Similarly, we say $\mathcal{H}$ is $\emph{$(n,d,\epsilon)$-evasive}$ if it is $(\mathcal{F},\epsilon)$-evasive.
\end{definition}

\begin{rem}
In \cref{defi_bounddeg}, we do not make any assumption about the dimension of the varieties in $\mathcal{F}$ or their irreducible components. We will see in \cref{sec_equi} that in fact, it suffices to consider the subfamily of equidimensional varieties or even irreducible varieties of dimension $n-k-1$ when constructing variety evasive $k$-subspace families.
\end{rem}

For $n,d\in\N^+$ and $k\in\{0,1,\dots,n-1\}$, define $N(k, d, n)$ by
%\[
%N(k,d,n):=\min\left\{ \binom{(k+1)(n+1+d)}{(k+1)d}, \binom{(n-k)(n+1+d)}{(n-k)d}, \binom{(d-1)(n+1+d)}{(d-1)d} \right\}.
%\]
\[
N(k,d,n):=\min\left\{ \binom{(k'+1)(n'+1+d)}{(k'+1)d}, \binom{(n-k)(n'+1+d)}{(n-k)d} \right\}.
\]
where $k':=\min\{k, d-2\}\leq k$ and $n':=k'+n-k\leq n$.

Our main theorem then states as follows.

\begin{theorem}[Main Theorem]\label{thm_main}
For  $n,d\in\N^+$, $k\in\{0,1,\dots,n-1\}$, and $\epsilon\in (0,1)$, there exists an $(n, d, \epsilon)$-evasive $k$-subspace family (resp.~affine $k$-subspace family) $\mathcal{H}$ on $\proj^n$ (resp.~$\aff^n$) of size $\poly(N(k,d,n), n, 1/\epsilon)$, 
which is bounded by $\poly(n^{\min\{k+1, n-k, d\}d}, d, 1/\epsilon)$.  
Moreover:
\begin{itemize}
\item If $\mathrm{char}(\F)=0$, where $\F$ denotes the base field, then the linear equations defining the projective or affine subspaces in $\mathcal{H}$ are defined over $\Q$.
Moreover,  the bit-lengths of the numerators and denominators of the coefficients of these linear equations are polynomial in $|\mathcal{H}|$, and the total bit complexity of computing these linear equations is polynomial in $|\mathcal{H}|$.
\item If $\mathrm{char}(\F)=p>0$,
 then the linear equations defining the projective or affine subspaces in $\mathcal{H}$ are defined over an extension field $\F_q$ of $\F_p$, where $q\leq \poly(|\mathcal{H}|, p)$.
The total bit complexity of computing these linear equations and constructing the field $\F_q$ is polynomial in $|\mathcal{H}|$ and $\log p$.
% and the coefficients of these linear equations are $\poly(|\mathcal{H}|)$-bounded.
\end{itemize}
In particular,   when $d$ is bounded, the bit complexity of constructing $\mathcal{H}$ is polynomial in $n/\eps$ (and $\log p$ if $\mathrm{char}(\F)=p>0$).
\end{theorem}

\begin{rem}
The two items in the above theorem bound the complexity of the coefficients that define $\mathcal{H}$. The same bounds apply to the coefficients in all constructions presented in this paper, and in particular, to those in \cref{thm_dnoetherproj} and \cref{thm_dnoetheraff} below. These bounds are needed for bounding the bit complexity of the construction of $\mathcal{H}$, which is crucial for demonstrating the explicitness of $\mathcal{H}$. We also remark that if we do not impose any restrictions on the complexity of the coefficients, then it is, in fact, straightforward to construct hitting sets of polynomial size unconditionally \cite[Lemma~4.2]{HS80}. This explains why we consider bit complexity as the complexity measure rather than assuming that each field operation takes unit cost, which is common in arithmetic complexity.
\end{rem}

\begin{rem}
A previous version of this paper \cite{Guo21} proved a weaker upper bound where $n'$ in the definition of $N(k,d,n)$ is replaced by $n$.
Our new bound in \cref{thm_main} has the advantage that when $n-k$ is small, we can get a subspace family of size $\poly(n,1/\epsilon)$ even if $d$ grows slowly in $n$:
\begin{itemize}
\item  As $N(k,d,n)\leq  \binom{(k'+1)(n'+1+d)}{(k'+1)d}\leq \binom{(d-1)(n-k+2d-1)}{(d-1)d}$, we can afford any $d\leq f(n)$ for some $f(n)=\omega_n(1)$ when $n-k=n^{o(1)}$.
\item Similarly, by the bound $N(k,d,n)\leq \binom{(n-k)(n'+1+d)}{(n-k)d}  \leq \binom{(n-k)(n-k+2d-1)}{(n-k)d}$, we can afford any $d=O(\log n)$ when $n-k=O(1)$.
\end{itemize}

\end{rem}

%For convenience, we say a field element $\alpha\in\F$ is \emph{$N$-bounded} for some $N>0$ if either  $\mathrm{char}(\F)=0$ and $\alpha$ is a rational number of bit-length  at most $N$, or  $\mathrm{char}(\F)=p>0$ and $\alpha$ lives in a finite extension of $\F_p$ of cardinality at most $\max\{N, p\}$. 

\paragraph{Lower bound.}
As a complement of the above result, we establish the following \emph{tight} lower bound for  projective   $k$-subspace families. It implies that when $n-k=n^{\Omega(1)}$, the assumption of $d$ being bounded is necessary for a projective $(n, d)$-evasive $k$-subspace family to have polynomial size.

\begin{restatable}{theorem}{lb}\label{thm_lb}
Let  $n,d\in\N^+$ and $k\in\{0,1,\dots,n-1\}$. 
Let $\mathcal{F}$ be the family of equidimensional projective subvarieties of $\proj^n$ of dimension $n-k-1$ and degree at most $d$.
Suppose $\mathcal{H}$ is an $\mathcal{F}$-evasive $k$-subspace family on  $\proj^n$. 
%(This holds in particular if $\mathcal{H}$ is $(n, d, \epsilon)$-evasive for some $\epsilon>0$.)
Then 
\[
|\mathcal{H}|\geq \begin{cases} (n-k)(k+1)+1 & \text{if } d=1, \\
 \max\left\{d(n-k)(k+1)+1, \binom{d+n-k}{d}+(n-k+1)k   \right\} & \text{if } d>1.
\end{cases}
\]
In particular, $|\mathcal{H}|$ is superpolynomial in $n$ when $n-k=n^{\Omega(1)}$ and $d=\omega(1)$.
\end{restatable}

\begin{rem}[Tightness of the lower bound]
When $d=1$, the lower bound $|\mathcal{H}|\geq (n-k)(k+1)+1$ in \cref{thm_lb} is achieved by known explicit lossless rank condensers \cite{FS12, FSS14, For14} (see \cref{sec_condenser}). For general $d$, the lower bound  in \cref{thm_lb}  is also tight and matched by \emph{non-explicit} constructions. See \cref{sec_lb} for a discussion.

In general, there is a gap between known upper bounds from explicit constructions and the tight lower bound. In particular, when $d\leq (n-k)^{1-\delta}$ for some constant $\delta>0$, our lower bound gives $|\mathcal{H}|\geq (n-k)^{\Omega(d)}+\poly(n)$ while the upper bound in \cref{thm_main} gives $|\mathcal{H}|\leq (n-k)^{O(d\min\{k,d\})}+\poly(n)$.
\end{rem}

Next, we list two applications of our Main Theorem (\cref{thm_main}): derandomizing Noether's normalization lemma for varieties of low degree, and polynomial identity testing for a special family of depth-4 arithmetic circuits.

\subsubsection{Derandomizing Noether's Normalization Lemma}

\emph{Noether's normalization lemma}, introduced by Noether \cite{Noe26}, is an important result in commutative algebra and algebraic geometry with many applications. For example, it is used in the development of dimension theory and can be used to prove Grothendieck's generic freeness lemma \cite{Eis13}. It also has applications in computational algebraic geometry, e.g., computing the dimension of a projective variety \cite{GH93, GHLMS00}.

The usual geometric formulation of Noether's normalization lemma states that for any affine variety
%\footnote{To be more precise, $\mathcal{V}$ can be a \emph{closed subscheme} of $\aff^n$, which is more general than an affine subvariety as it can be \emph{nonreduced}. We ignore the difference here to keep it elementary.}  
$\mathcal{V}\subseteq \aff^n$ of dimension $r$, there exists a surjective \emph{finite morphism} $\pi: \mathcal{V}\to \aff^r$.
(See \cref{sec_ag} for the definition of finite morphisms.) Moreover, $\pi$ may be chosen to be the restriction of a linear map $\aff^n\to\aff^r$.\footnote{For simplicity, we assume the base field is algebraically closed and hence infinite. But the lemma and our derandomization are valid as long as the field is large enough, depending on the variety $\mathcal{V}$. We remark that Nagata \cite{Nag62} proved a version of the normalization lemma that is deterministic and does not require the base field to be sufficiently large, but the morphism he used is highly nonlinear.
% and is defined by polynomials of degree double exponential in $n$. 
Moreover, Nagata's normalization lemma crucially relies on the fact that the variety is affine, while the normalization lemma we consider here extends to the projective case.
%Bruce and Erman \cite{BE19} proved an effective Noether normalization lemma over finite fields, which states that with high probability, a random tuple of degree-$d$ polynomials over a finite field induces a valid finite morphism for large enough $d$ satisfying a certain effective bound. %We leave it as an open problem to derandomize the normalization lemma that makes use of such random nonlinear polynomials.
} 
 There is also a related projective or graded version of the lemma, which states that for any projective variety  $\mathcal{V}$ of dimension $r$, there exists a surjective finite morphism $\pi: \mathcal{V}\to \proj^r$. A special form of this lemma goes back to Hilbert \cite{Hil90}.  

In these versions of Noether's normalization lemma, it can be shown that with high probability, a random linear map  yields a valid finite morphism $\pi$, where ``random'' means the coefficients of the linear map are chosen randomly from a sufficiently large finite set $S\subseteq\F$. It is thus a natural question to derandomize the lemma.

Mulmuley \cite{Mul17} studied a form of  Noether's normalization lemma and proved that derandomizing it is equivalent to a strengthened form of the black-box derandomization of PIT. 
There, the ambient projective space has exponential dimension and the problem is constructing a finite morphism $\pi: \mathcal{V}\to \proj^k$ with a \emph{succinct} specification in deterministic polynomial time, where $k=\poly(\dim(\mathcal{V}))$ and $\mathcal{V}$ is an \emph{explicit variety}  \cite{Mul17}.  This problem was later shown to be in PSPACE \cite{FS17, GSS18}. The special case for the ring of matrix invariants under simultaneous conjugation was solved in quasipolynomial time by Forbes and Shpilka \cite{FS13}.  
%In one concrete example (regarding the variety $\mathcal{V}=\Delta[\mathrm{det},m]$  \cite{Mul17}), the ambient projective space is $\proj \mathcal{X}$ where $\mathcal{X}$ is the linear space of degree-$m$ homogeneous polynomials in $r=m^2$ variables. Thus, the dimension of $\proj \mathcal{X}$ is exponential in $m$. The variety $\mathcal{V}\subseteq\proj \mathcal{X}$   have dimension $\poly(m)$, and the problem of derandomizing Noether's normalization lemma, as formulated in \cite{Mul17}, asks about finding  $\mathbf{x}_1,\dots, \mathbf{x}_k\in \F^r$ in time $\poly(m)$, where $k=\poly(m)$, such that the linear map $P\mapsto (P(\mathbf{x}_1), \dots, P(\mathbf{x}_k))$ from $\mathcal{X}$ to $\F^k$
%induces a finite morphism $\mathcal{V}\to \proj^{k-1}$. 

We consider Noether's normalization lemma in its original context and completely derandomize it for projective/affine varieties of bounded degree. The following two theorems summarize our results.

\begin{restatable}{theorem}{dnoetherproj}\label{thm_dnoetherproj}
Let $n, d\in\N^+$, $r\in \{0,1,\dots, n\}$, $k=n-r-1$, and $\epsilon\in (0,1)$.
There exists an explicit collection $\mathcal{L}$ of linear maps $\aff^{n+1}\to\aff^{r+1}$ of size $\poly(N(k,d,n), n, 1/\epsilon)$
such that for every subvariety $\mathcal{V}\subseteq\proj^n$ of dimension $r$ and degree at most $d$,
all but at most $\epsilon$-fraction of $\pi\in\mathcal{L}$ induce a surjective finite morphism from $\mathcal{V}$ to $\proj^{r}$.\footnote{Let $N(k,d,n)=1$ when $r=n$ (i.e., $k=-1$). Similarly, in \cref{thm_dnoetheraff}, let $N(k,d,n-1)=1$ when $r=n$ or $r=0$ (i.e., $k=-1$ or $k=n-1$).}
Moreover, $\mathcal{L}$ can be computed in time polynomial in $|\mathcal{L}|$ (and $\log p$, if $\mathrm{char}(\F)=p>0$).
\end{restatable}

\begin{restatable}{theorem}{dnoetheraff}\label{thm_dnoetheraff}
Let $n, d\in\N^+$, $r\in \{0,1,\dots, n\}$, $k=n-r-1$, and $\epsilon\in (0,1)$.
There exists an explicit collection $\mathcal{L}$ of linear maps $\aff^{n}\to\aff^{r}$ of size $\poly(N(k,d,n-1), n, 1/\epsilon)$
such that for every subvariety $\mathcal{V}\subseteq\aff^n$ of dimension $r$ and degree at most $d$,
all but at most $\epsilon$-fraction of $\pi\in\mathcal{L}$ restrict to a surjective finite morphism from $\mathcal{V}$ to $\aff^{r}$.
Moreover, $\mathcal{L}$ can be computed in time polynomial in $|\mathcal{L}|$ (and $\log p$, if $\mathrm{char}(\F)=p>0$).
\end{restatable}

%\cref{thm_dnoetherproj} is proved by derandomizing a standard proof of Noether's normalization lemma that has a geometric flavor \cite{Sha13}. Namely, we consider a projection $\pi: \proj^n \setminus W\to \proj^r$ sending $\mathbf{x}$ to $(\ell_1(\mathbf{x}),\cdots, \ell_{r+1}(\mathbf{x}))$, where $\ell_1,\dots,\ell_r$ are linear forms and $W$ is the $(n-r-1)$-subspace where these linear forms simultaneously vanish. It is known that $\pi$ restricts to a finite morphism $\mathcal{V}\to \proj^r$ iff $W\cap \mathcal{V}=\emptyset$. So the problem reduces to choosing a family of $(n-r-1)$-subspaces of $\proj^n$ such that most of them are disjoint from $\mathcal{V}$. This is exactly the property satisfied by our explicit variety evasive subspace families.
%
%\cref{thm_dnoetheraff} is proved similarly. Here  $\aff^n$ is viewed as an open subset of $\proj^n$ whose complement is the ``hyperplane at infinity'' $H_\infty$. Then we first construct a projection $\pi: \proj^n \setminus W\to \proj^r$ such that $W$ is a subspace of $H_\infty$ and is disjoint from the \emph{projective closure} of $\mathcal{V}$.
%Then restrict $\pi$ to $\aff^n$. By carefully choosing $\pi$, we can make sure that the restriction is a linear map $\aff^n\to \aff^r$ and is a surjective finite morphism.

\paragraph{Dimension-preserving morphisms vs. finite morphisms.} 
Our construction of finite linear morphisms preserves the dimension of a variety of low degree while reducing the dimension of the ambient space.
This generalizes the property of lossless rank condensers.
However, for the dimension-preserving property, better constructions are known.
For example, it can be shown that most of the linear maps $\aff^n\to \aff^t$  from a lossless rank condenser $\mathcal{H}\subseteq \F^{t\times n}$ already preserve the dimension of a variety $\mathcal{V}\subseteq \aff^n$. 
The intuition here is that $\mathcal{V}$ can be locally approximated at a nonsingular point $p\in\mathcal{V}$ by its \emph{tangent space} at $p$. (Note that such a nonsingular point always exists when $\mathcal{V}$ is a nonempty variety.) So any linear map that preserves the dimension of this tangent space also preserves the dimension of $\mathcal{V}$. 
%The existence of a nonsingular point of a \emph{variety} follows from the Jacobian criterion for smoothness \cite{Eis13}.
%This  was  used by Dvir \cite{Dvi12} in his explicit constructions of \emph{extractors for varieties}, which generalize \emph{affine extractors} \cite{GR08}. 

On the other hand, the morphisms we construct are \emph{finite morphisms}, which are strictly stronger than morphisms that are dimension-preserving.
In particular, a finite morphism $\pi$  always maps a closed set onto a closed set in the Zariski topology. Moreover, the preimage $\pi^{-1}(p)$ of \emph{every} point $p$ in the image of $\pi$ is a finite set.
%\footnote{In the context of \cref{thm_dnoetherproj} and \cref{thm_dnoetheraff}, the cardinality of this set is bounded by $\deg(\mathcal{V})$.}. 
 Neither of these two properties is necessarily satisfied by morphisms that are only dimension-preserving.
 
These properties of finite morphisms may be useful in extractor theory or other areas.
For example, in \cref{thm_dnoetheraff}, the cardinality of $\pi^{-1}(p)$ is bounded by the degree of $\mathcal{V}$ for every $p\in\pi(\mathcal{V})$, which translates into a lower bound for the min-entropy of the output of $\pi$ when the input random source is distributed over the variety $\mathcal{V}$.

\subsubsection{Depth-4 Polynomial Identity Testing}

Depth-4 arithmetic circuits, also known as $\Sigma\Pi\Sigma\Pi$ circuits, play a very important role in polynomial identity testing. In a surprising result, Agrawal and Vinay  \cite{AV08} proved that a complete derandomization of black-box PIT for depth-4 circuits implies an $n^{O(\log n)}$-time derandomization of PIT for general circuits of $\poly(n)$ degree.

Dvir and Shpilka  \cite{DS07} initialized the approach of applying \emph{Sylvester--Gallai} type theorems in geometry to PIT for depth-3 ($\Sigma\Pi\Sigma$) circuits.
Extending this approach, Gupta \cite{Gup14} formulated a  conjecture of Sylvester--Gallai type and proved that  his conjecture implies a complete derandomization of black-box PIT for  depth-4 circuits with bounded top fan-in and bottom fan-in (also called $\Sigma\Pi\Sigma\Pi(k, r)$ circuits, where $k, r=O(1)$). 
Peleg and Shpilka \cite{Shp19, PS20, PS20_2} proved that this conjecture holds for $k=3$ and $r=2$, and used it to give a polynomial-time black-box PIT algorithm for $\Sigma\Pi\Sigma\Pi(3, 2)$ circuits.
Using a different approach, Dutta, Dwivedi, and Saxena \cite{DDS21} gave a quasipolynomial-time black-box PIT algorithm for $\Sigma\Pi\Sigma\Pi(k, r)$ circuits.
Finally, in an exciting breakthrough, Limaye, Srinivasan, and Tavenas \cite{LST21} obtained a subexpoential-time black-box PIT algorithm for all arithmetic circuits of bounded depth.

In \cite{Gup14}, Gupta divided $\Sigma\Pi\Sigma\Pi(k, r)$ into two families: those in a certain Sylvester--Gallai configuration and those that are not. His conjecture states that the circuits in the first family always have bounded \emph{transcendence degree}, depending only on $k$ and $r$. If the conjecture is true, then the results in \cite{BMS13, ASSS16} imply a complete derandomization of the black-box PIT for this family.
 For the second family of circuits, which we call \emph{non-SG} circuits, he proved that the black-box PIT can also be derandomized completely.
 
 \begin{theorem}[{\cite{Gup14}}]\label{thm_gupta}
There exists a deterministic black-box PIT algorithm with time complexity $(dnk)^{\poly(r^{k^2}+k)}$  for non-SG $\Sigma\Pi\Sigma\Pi(k, r)$ circuits of degree at most $d$ in $X_1,\dots,X_n$ over $\C$.
In particular, the algorithm runs in polynomial time when $k$ and $r$ are bounded.
\end{theorem}

Gupta's proof of \cref{thm_gupta} is quite complex and used tools from computational algebraic geometry, including an effective version of Bertini irreducibility theorem \cite{HS81} and radical  membership testing (which in turn depends on \emph{effective Nullstellensatz} \cite{Kol88, Dub93}).

We observe that what is needed here is simply an explicit construction of subspaces intersecting certain varieties with (at most) the expected dimension.
Plugging in our explicit construction of variety evasive subspace families, we obtain an improved black-box PIT algorithm with a simple proof.

\begin{restatable}{theorem}{pit}\label{thm_pit}
There exists a deterministic black-box PIT algorithm with time complexity polynomial in $d\cdot\binom{k(n+1+r^k)}{kr^k}\cdot\binom{k-1+d}{k-1}\leq \poly(d^{k}, n^{r^k}, r^{k^2r^k})$ (and $\log p$, if $\mathrm{char}(\F)=p>0$)  for non-SG  $\Sigma\Pi\Sigma\Pi(k, r)$ circuits of degree at most $d$ in $X_1,\dots,X_n$ over an algebraically closed field $\F$.
\end{restatable}

In particular, \cref{thm_pit} improves the exponent of $n$ in the time complexity from $\poly(r^{k^2}+k)$ to $O(r^k)$, and the exponent of $d$ from $\poly(r^{k^2}+k)$ to $O(k)$. Moreover, our proof is more direct and conceptually simpler than the proof in \cite{Gup14}.

\begin{rem}
In \cite{Muk16}, Mukhopadhyay gave a deterministic polynomial-time black-box PIT algorithm for $\Sigma\Pi\Sigma\Pi(k, r)$ circuits  satisfying a variant of the non-SG assumption. (Its time complexity is similar to the time complexity in \cref{thm_gupta}.)
%The time complexity is bounded by $(dn)^{r^{k^2}}$.
It appears to us that his assumption in fact implies the non-SG assumption.
The main tool used there is the \emph{multivariate resultant}, which may be related to our approach  based on Chow forms (see \cref{sec_overview}). Indeed, it is known that a multivariate resultant is the Chow form of a Veronese variety \cite[Chapter~3, Example~2.4]{GKZ94}.
\end{rem}

\subsection{Proof Overview}\label{sec_overview}

We present an overview of our proof of \cref{thm_main} and that of \cref{thm_lb}.

\paragraph{Overview of the proof of \cref{thm_main}.} In the proof of \cref{thm_main}, we focus on constructing a $k$-subspace family on $\proj^n$. The  case of $\aff^n$ can be easily derived from it by viewing $\aff^n$ as an open subset of $\proj^n$ and restricting to this subset.

Consider a variety  $\mathcal{V}\subseteq\proj^n$ of degree at most $d$. 
We want to construct  a $k$-subspace family $\mathcal{H}$ on $\proj^n$, independent of $\mathcal{V}$, such that  all but at most $\epsilon$-fraction of $W\in\mathcal{H}$ evade $\mathcal{V}$. Our key ideas can be summarized as follows.

\subparagraph{Reducing to the equidimensional/irreducible case of dimension $n-k-1$.} As a first step, we reduce the problem to the special case that $\mathcal{V}$ is an \emph{equidimensional} (or even irreducible) variety of $\proj^n$ of dimension $n-k-1$, which means every irreducible component of $\mathcal{V}$ has dimension exactly $n-k-1$. This step is explained in \cref{sec_equi}.

\subparagraph{Hitting the Chow form of $\mathcal{V}$.}
Denote by $\pgr(k, n)$ the Grassmannian consisting of of all $k$-subspaces of $\proj^n$.
As $\codim(\mathcal{V})=n-(n-k-1)>k$, a \emph{general}  $k$-subspace $W\in \pgr(k, n)$ is disjoint from $\mathcal{V}$, but we want to find such $W$ explicitly. 

One remarkable fact in algebraic geometry is that there is a single polynomial  $\widetilde{R}_\mathcal{V}$ on the Grassmannian $\pgr(k, n)$   that defines precisely the subset of $k$-subspaces that intersect $\mathcal{V}$. This polynomial  $\widetilde{R}_\mathcal{V}$  is called the \emph{Chow form} of $\mathcal{V}$ (in \emph{Stiefel coordinates}).
Chow forms are also known as \emph{associated forms}, \emph{Cayley forms}, or \emph{Cayley--van der Waerden--Chow forms} in literature. They were introduced by Cayley \cite{Cay60} to represent curves in $\proj^3$ and later generalized by Chow and van der Waerden \cite{CvdW37}. See \cite{DS95} for an  introduction to Chow forms and \cite{GKZ94} for an exposition in the context of elimination theory.

To be more specific, for a $k$-subspace $W\in \pgr(k, n)$, we choose a $(k+1)\times (n+1)$ matrix $A$ that represents $W$. The Chow form $\widetilde{R}_\mathcal{V}$ is a polynomial of degree $(k+1)\deg(\mathcal{V})$ in $(k+1)  (n+1)$ variables    with the following property:  $\widetilde{R}_\mathcal{V}$ vanishes at  the matrix $A$ (viewed as a list of $(k+1)  (n+1)$ coordinates) if and only if
$\mathcal{V}\cap W\neq\emptyset$.
Thus, $\widetilde{R}_\mathcal{V}$ defines precisely the subset of ``bad'' $k$-subspaces that we want to avoid.

Therefore, the problem becomes finding a collection of $(k+1)\times (n+1)$ matrices of full rank that ``hit''  the polynomial $\widetilde{R}_\mathcal{V}$ of degree  $(k+1)\deg(\mathcal{V})\leq (k+1)d$.
Using black-box PIT for low degree polynomials  (see \cref{sec_lowdeg}),  we are able to construct an $(n, d, \epsilon)$-evasive  $k$-subspace family of size polynomial in $\binom{(k+1)(n+1+d)}{(k+1)d}$ and $1/\epsilon$, which is $\poly(n, 1/\epsilon)$ when $k$ and $d$ are both bounded.
A similar ``dual'' construction yields a $k$-subspace family of size polynomial in $\binom{(n-k)(n+1+d)}{(n-k)d}$ and $1/\epsilon$, which is $\poly(n, 1/\epsilon)$ when both $n-k$ and $d$ are bounded.
For applications where $d$ is small and either $k$ or $n-k$ is small (e.g., \cref{thm_pit}), these constructions are good enough. However, when $k$ and $n-k$ are both linear in $n$, the resulting $k$-subspace families have exponential size in $n$, even if $d$ is bounded.

\subparagraph{A two-step construction.} 
To obtain a good construction for \emph{arbitrary} dimension $k$, we use a connection with Noether normalization.
It is a standard fact in algebraic geometry that a subspace $W\subseteq \proj^n$ is disjoint from a variety $\mathcal{V}$ iff it defines a \emph{projection} $\proj^n\setminus W\to \proj^{n-\dim(W)-1}$ that restricts to a finite morphism from $\mathcal{V}$ to $\proj^{n-\dim(W)-1}$.
Thus, we may reformulate our problem as finding finite morphisms $\pi: \mathcal{V}\to \proj^{n-k-1}$ that come from  projections.

We also need another fact, which states that   the codimension of an irreducible  subvariety $\mathcal{V}\subseteq \proj^n$ in $\spn(\mathcal{V})$ is at most $\deg(\mathcal{V})-1$, where $\spn(\mathcal{V})$ denotes the smallest projective subspace containing $\mathcal{V}$ (see \cref{lem_span}).
Therefore, for irreducible $\mathcal{V}$ of degree at most $d$, there exists a projective subspace $\Lambda$ of dimension (at most) $\dim(\mathcal{V})+d-1$ that contains $\mathcal{V}$. 

Our idea is to use a two-step construction. Namely, we first construct a finite morphism $\pi_1: \mathcal{V}\to \proj^{\dim (\Lambda)}$ that comes from a projection, and then construct another finite morphism $\pi_2: \mathcal{V}'\to \proj^{n-k-1}$, where $\mathcal{V}':=\pi_1(\mathcal{V})\subseteq  \proj^{\dim (\Lambda)}$. Composing $\pi_1$ with $\pi_2$ yields a desired map from $\mathcal{V}$ to $\proj^{n-k-1}$.

The first step is just the problem of constructing lossless rank condensers,  which has an optimal solution \cite{FS12, FSS14} (see \cref{sec_condenser}).  For the second step, we need to hit the Chow form of $\mathcal{V}'$. 
Thanks to the first step, the codimension of $\mathcal{V}'$ in  $\proj^{\dim(\Lambda)}$ is only
$\dim(\Lambda)-\dim(\mathcal{V}')=\dim(\Lambda)-\dim(\mathcal{V})\leq d-1$.  As the codimension is low, we may use black-box PIT for low degree polynomials just like before, and \cref{thm_main} follows.\footnote{A preliminary version of this paper \cite{Guo21} used a similar two-step construction but did not exploit the connection with Noether normalization. It is more redundant and yields a somewhat weaker  result than  \cref{thm_main}.}

 Finally,  the above connection with Noether normalization also allows us to derive \cref{thm_dnoetherproj} and \cref{thm_dnoetheraff} easily from \cref{thm_main}.

\paragraph{Overview of the proof of \cref{thm_lb}.}
Our lower bound (\cref{thm_lb}) follows from a dimension counting argument.
Let $C(r, d, n)$ be the set of  all varieties $\mathcal{V}\subseteq\proj^n$ of dimension $r:=n-k-1$ and degree $d$, which is the set of varieties that we want to evade.

Roughly speaking,  the idea is to show that (1) $C(r, d, n)$ itself can be realized as a subvariety of some projective space $\proj^N$, and (2) for every $k$-subspace $W$, the subset of $\mathcal{V}\in C(r,d,n)$ that $W$ fails to evade is the intersection of $C(r,d,n)$ with some hyperplane $H_W$ of $\proj^N$.

To see how (1) and (2) above lead to a lower bound, suppose $\mathcal{H}$ is a  $C(r, d, n)$-evasive $k$-subspace family, i.e.,  for any $\mathcal{V}\in C(r, d, n)$, there exists  $W\in\mathcal{H}$ that is disjoint from $\mathcal{V}$. Then the intersection $C(r, d, n)\cap \bigcap_{W\in\mathcal{H}} H_W$ must be empty. On the other hand, taking the intersection with each hyperplane $H_W$ reduces the dimension  of a projective variety by at most one. So we have a lower bound $|\mathcal{H}|\geq \dim(C(r, d, n))+1$.

How do we realize $C(r, d, n)$ as a subvariety of $\proj^N$? It turns out that this is a classical problem in the study of moduli spaces  and a solution was given by Cayley \cite{Cay60} and Chow--van der Waerden  \cite{CvdW37} using the \emph{Chow embedding}:
The Chow embedding $C(r, d, n)\to \proj^N$ simply sends a variety $\mathcal{V}$ to  its Chow form $\widetilde{R}_\mathcal{V}$, 
where $\widetilde{R}_\mathcal{V}$ is viewed as a point in the projective space $\proj^N$ whose homogeneous coordinates are given by the coefficients of $\widetilde{R}_\mathcal{V}$.\footnote{The actual Chow embedding we use has a slightly different form, which is essentially equivalent to the one described here.}

A technical issue here is that the image of $C(r, d, n)$ under the Chow embedding is generally not  closed in the Zariski topology. To fix this issue, the definition of $C(r, d, n)$ needs to be modified so that it contains not only  subvarieties of $\proj^n$, but also   \emph{(effective) algebraic cycles} on $\proj^n$, which are a generalization of subvarieties. 
%(As a concrete example, the Chow forms of subvarieties of degree $d$ and dimension $n-1$ are squarefree degree-$d$ polynomials, while the Chow forms of effective $r$-cycles of degree $d$ give all degree-$d$ polynomials.)
 A theorem of Chow and van der Waerden \cite{CvdW37} then states that the Chow embedding does embed $C(r,d,n)$ in a projective subspace $\proj^N$ as a subvariety,  known as a \emph{Chow variety}.
 
Finally, we also need a lower bound for the dimension of the Chow variety $C(r, d, n)$. In fact, the exact value of $\dim(C(r,d,n))$ was determined by   Azcue \cite{Azc93} and independently by Lehmann \cite{Leh17}. Plugging in the value of $\dim(C(r,d,n))$ proves \cref{thm_lb}.

\subsection{Other Related Work}

 In \cite{DKL14}, Dvir, Koll{\'a}r,  and Lovett constructed explicit \emph{variety evasive sets}, which are large subsets of $\F_q^n$ over a finite field $\F_q$ that have small intersection with affine varieties of fixed dimension and bounded degree.
It generalizes an earlier construction of \emph{subspace evasive sets} of Dvir and Lovett \cite{DL12}.
The definition of evasiveness there is different from ours, but they are related, since a key step in the proofs of \cite{DL12, DKL14} is proving the intersection of two varieties has  dimension zero.
We also note that a  subspace/variety evasive set   is a single set, defined in a highly nonlinear way, whereas  we define a variety evasive subspace family to be a collection of projective or affine subspaces. Finally, the results in \cite{DL12, DKL14} hold only for affine subspaces/subvarieties, whereas we give our construction first in the projective setting and then derive the affine counterpart from it.

Guruswami and Xing in \cite{GX13} introduced a related notion called \emph{subspace designs}. A   subspace design is a collection $\mathcal{H}$ of large subspaces of $\F^n$ such that for any small subspace $V\subseteq\F^n$, the number of $W\in\mathcal{H}$ satisfying $\dim(W\cap V)>0$ is small  (or even the sum $\sum_{W\in\mathcal{H}}\dim(W\cap V)$ is small).
An equivalence between subspace designs and lossless rank condensers was proved in \cite{FG15}.
Explicit subspace designs were constructed by Guruswami and Kopparty \cite{GK16} and also by Guruswami, Xing, and Yuan \cite{GXY18}. 
They have applications to constructing explicit list-decodable codes with small list size \cite{GX13, GWX16, KRSW18, GR20} and explicit dimension expanders \cite{FG15, GRX21}. Subspace designs were also used to prove lower bounds in communication complexity \cite{CGS20}. 

 Jeronimo, Krick, Sabia, and Sombra \cite{JKSS04} gave a randomized algorithm, in the Blum-Shub-Smale  model over fields of characteristic zero, that computes the Chow forms of varieties defined by input polynomials. The (expected) time complexity of their algorithm is polynomial in the sizes of the arithmetic circuits encoding the input polynomials and the \emph{geometric degree}  of the polynomial system. See also the survey by Krick \cite{Kri02}.
 
 Chow varieties of effective zero-cycles and their higher secant varieties are related to lower bounds for depth-3 arithmetic circuits. They have received a considerable amount of attention in Geometric Complexity Theory  \cite{Lan12, Lan15}.

\paragraph{Organization of the paper.}

Preliminaries and notations are given in \cref{sec_pre}. We prove \cref{thm_main}, \cref{thm_dnoetherproj}, and \cref{thm_dnoetheraff} in \cref{sec_main}.  
In \cref{sec_lb}, we prove the lower bound (\cref{thm_lb}) and also give a non-explicit construction that matches this lower bound. The application to PIT for depth-4 circuits (\cref{thm_pit}) is explained in \cref{sec_app}.
Finally, we list some open problems and future directions in \cref{sec_open}.

\section{Preliminaries and Notations}\label{sec_pre}

Define $\N:=\{0,1,2\dots\}$ and $\N^+:=\{1,2,\dots\}$. Let $[n]:=\{1,2,\dots, n\}$ for $n\in\N$.
%The cardinality of a set $S$ is denoted by $|S|$.
For a set $S$ and $k\in\N$, denote by $\binom{S}{k}$ the set of all subsets of $S$ of cardinality $k$.
%Denote the algebraic closure of a field $\K$ by $\overline{\K}$.

Denote by $\F$ an algebraically closed field throughout this paper.  We use notations like $\F[X_{i,j}: i\in [n], j\in [m]]$ to denote the polynomial ring over $\F$ in a finite set of variables (in this case, in the set of variables $\{X_{i,j}: i\in [n], j\in [m]\}$).
The vector space of $n\times m$ matrices over $\F$ is denoted by $\F^{n\times m}$.

For an $n\times m$ matrix $A$  and subsets $S\subseteq [n]$, $T\subseteq [m]$, denote by $A_{S,T}$ the submatrix of $A$ whose rows and columns are selected by $S$ and $T$ respectively, where the orderings of rows and columns are preserved. 

%For an $n\times m$ matrix $A$ and an $m\times n$ matrix $B$, the \emph{Cauchy--Binet formula} states that
%\[
%\det(AB)=\sum_{S\in \binom{[m]}{n}} \det(A_{[n], S})\det(B_{S, [n]}).
%\]

\subsection{Black-Box PIT for Low Degree Polynomials}\label{sec_lowdeg}

For convenience, we strengthen the definition of hitting sets as follows.

\begin{definition}[$\epsilon$-hitting set]
Let $\mathcal{F}$ be a family of polynomials in $\F[X_1,\dots,X_n]$ and $\epsilon\in (0,1)$.
We say a finite collections of points $\mathcal{H}\subseteq \F^n$ is an \emph{$\epsilon$-hitting set} for $\mathcal{F}$
if for any nonzero $Q\in\mathcal{F}$, the evaluation $Q(\alpha)$ is nonzero for all but at most $\epsilon$-fraction of $\alpha\in\mathcal{H}$.
\end{definition}

We need an explicit construction of $\epsilon$-hitting sets for low degree polynomials. 
This problem has been well studied \cite{DL78, Zip79, Sch80, KS01, Bog05, Lu12, CT13, Bsh14, BP20}. For completeness,   we present a construction   based on sparse polynomial identity testing.

Recall that a polynomial is \emph{$s$-sparse} if it has at most $s$ monomials. We need the following lemma from \cite{AGKS15}.

\begin{lemma}[{\cite[Lemma~4, restated]{AGKS15}}]\label{lem_sparse}
For $n,s,d\in\N^+$ and $\epsilon_0\in (0,1)$,
there exist maps $w_1,w_2,\dots, w_N: [n]\to [N\log N]$, where $N=\poly(n, s, \log d, \epsilon_0^{-1})$,  such that for any nonzero $s$-sparse polynomial $f\in\F[X_1,\dots, X_n]$ of individual degree at most $d$, all but at most $\epsilon_0$-fraction of  $w_i$ among $w_1,w_2,\dots,w_N$  satisfies $f(Y^{w_i(1)},\dots,Y^{w_i(n)})\neq 0$.
Moreover, the bit complexity of computing $w_1,w_2,\dots,w_N$ is polynomial in $N$.
\end{lemma}

Given $n,d\in\N^+$ and $\epsilon\in (0,1)$, we construct an $\epsilon$-hitting set for $n$-variate polynomials of degree at most $d$ as follows:
\begin{enumerate}
\item Let $s=\binom{n+d}{d}$, $\epsilon_0=\epsilon/2$, and $M=\lceil \epsilon_0^{-1} dN\log N\rceil$, where $N$ is as in \cref{lem_sparse}.
\item Let $w_1,\dots,w_N$ be as in \cref{lem_sparse}, which can be computed in time $\poly(N)$. 
\item If $\mathrm{char}(\F)=0$, let $S=[M]\subseteq \Z\subseteq\F$.  If $\mathrm{char}(\F)=p>0$, choose the smallest $p$-power $q$ such that $q\geq M$, and choose $S$ to be a subset of $\F_q\subseteq \F$  of cardinality $M$.
We remark that $\F_q$ can be constructed deterministically in time $\poly(M, \log p)$. To see this, note that $q\leq Mp$ by the minimality of $q$. If $M\leq p$, then $\F_q$ is just $\F_p$. On the other hand, if $p<M\leq q$, then $\F_q$ can be constructed in time $\poly(p, [\F_q: \F_p])$ (see, e.g., \cite{Len90}), which is polynomial in $M$ since $p<M$ and $q\leq Mp$.
\item Finally, construct the following collection of points in $\F^n$ of size $MN$
\[
T=\{(\alpha^{w_i(1)},\dots,\alpha^{w_i(n)}): \alpha\in S, i\in [N]\}\subseteq \F^n.
\]
\end{enumerate}

\begin{lemma}\label{lem_lowdeg}
For any nonzero polynomial $f\in\F[X_1,\dots,X_n]$ of degree at most $d$, we have $f(u)\neq 0$ for all but at most $\epsilon$-fraction of  $u\in T$. The collection $T$ has cardinality $\poly\left(\binom{n+d}{d}, 1/\epsilon\right)$ and can be computed in time $\poly(|T|)$. 
\end{lemma}

\begin{proof}
Let $f\in\F[X_1,\dots,X_n]$ be a nonzero polynomial of degree at most $d$. Note that $f$ is trivially $s$-sparse, where $s=\binom{n+d}{d}$. So by \cref{lem_sparse}, for all but at most $\epsilon_0$-fraction of $i\in [N]$, we have $\widetilde{f}_i:=f(Y^{w_i(1)},\dots,Y^{w_i(n)})\neq 0$. Consider $i\in [N]$ such that $\widetilde{f}_i\neq 0$.
Note that $\widetilde{f}_i$ is a univariate polynomial of degree at most $d N\log N$. So it has at most $d N\log N\leq \epsilon_0 M$ zeros.
Therefore, by the choice of $M$, we have $f(\alpha^{w_i(1)},\dots,\alpha^{w_i(n)})=\widetilde{f}_i(\alpha)\neq 0$ for all but at most $\epsilon_0$-fraction of $\alpha\in S$. It follows that $f(u)\neq 0$ holds for all but at most $\epsilon$-fraction of  $u\in T$, as claimed. The rest of the lemma follows easily from the construction.
\end{proof}

Note that the seed length required to choose a random element in $T$ is $\log |T|=O(\log\binom{n+d}{d}+\log(1/\epsilon))$, which is optimal up to a constant factor. 
We have made no effort to optimize the constant hidden in $O(\cdot)$.
Interested readers may find the state-of-the-art result in \cite{BP20}, which achieves the optimal constant, at least when $d=n^{o(1)}$.

\subsection{Explicit Lossless Rank Condensers}\label{sec_condenser}

We need the following lemma  in the context of \emph{lossless rank condensers}. 
The construction in the lemma was given by Forbes and Shpilka \cite{FS12} and the lemma itself follows implicitly from the analysis of Forbes, Saptharishi, and Shpilka in \cite{FSS14}. It was also stated explicitly in \cite[Theorem~5.4.3]{For14}.

\begin{lemma}[{\cite{FSS14, For14}}]\label{lem_condenser}
Let $n\in \N^+$ and $r\in [n]$. Let $\omega\in \F^\times$ such that the multiplicative order of $\omega$ is at least $n$.
Define the $r\times n$ matrix $W=(w_{i,j})_{i\in [r], j\in [n]}$ over $\F[X]$ by
\[
w_{i,j}=(\omega^{i-1} X)^{j-1}.
\]
Then for every $n\times r$ matrix $M$ over $\F$ of rank $r$, the polynomial $\det(WM)\in\F[X]$ is  nonzero and has degree at most $r(n-r)$ after dividing out powers of $X$.
\end{lemma}

\begin{corollary}\label{cor_condenser}
Let $n, r, W$ be as in \cref{lem_condenser} and $\epsilon\in (0,1)$. Let $S\subseteq \F^\times$ be a finite set of cardinality at least $r(n-r)/\epsilon$. 
For every $n\times r$ matrix $M$ over $\F$ of rank $r$, we have $\mathrm{rank}(W(\alpha)M)=r$ for all but at most $\epsilon$-fraction of $\alpha\in S$, where $W(\alpha)$ denotes the matrix $(w_{i,j}(\alpha))_{i\in [r], j\in [n]}$ over $\F$. 
\end{corollary}

\cref{cor_condenser} states that the collection $\{W(\alpha): \alpha\in S\}$ of matrices is a (weak) \emph{$(r, \epsilon |S|)$-lossless rank condenser}, as defined in \cite{FG15}.
Note that for each $\alpha\in S$, we have  $\mathrm{rank}(W(\alpha))=r$  and hence $W(\alpha)$ corresponds to an $(r-1)$-subspace $U_{W(\alpha)}$ of $\proj^{n-1}$.
As explained in the introduction, the collection $\mathcal{H}=\{U_{W(\alpha)}: \alpha\in S\}$ is an $(\mathcal{F}, \epsilon)$-evasive $(r-1)$-subspace family on $\proj^{n-1}$, where $\mathcal{F}$ is the family of $(n-r-1)$-subspaces of $\proj^{n-1}$. Choosing $S$ of size $r(n-r)+1$ and $\epsilon=1-\frac{1}{r(n-r)+1}$ shows that the lower bound in \cref{thm_lb} is achieved when $d=1$.

\subsection{Preliminaries on Algebraic Geometry}\label{sec_ag}

We list basic preliminaries and notations on algebraic geometry used in this paper. One can also refer to a standard text, e.g., \cite{Sha13, Har13}.

\paragraph{Affine and projective spaces.}
For $n\in\N$, write $\aff^n$ for the \emph{affine $n$-space} over $\F$. It is defined to be the set $\F^n$ equipped with the {\em Zariski topology}, defined as follows: A subset $S\subseteq \aff^n$ is {\em (Zariski-)closed} if it is the set of common zeros of a set of polynomials in $\F[X_1,\dots,X_n]$. 
%For other subsets $S$ it makes sense to consider the {\em closure} $\ol{S}$-- the smallest closed set containing $S$. 
%Set $S$ is {\em dense} in $\aff^d$  if $\ol{S}=\aff^d$.  
The complement of a closed set is an {\em open} set. The origin of an affine space is denoted by $\mathbf{0}$.

Write $\proj^n$ for the \emph{(projective) $n$-space} over $\F$, defined to be the quotient set $(\aff^{n+1}\setminus\{\mathbf{0}\})/\sim$, where $\sim$ is the equivalence relation defined by scaling, i.e., $u\sim v$ if $u=cv$ for some $c\in\F^\times$. The set $\proj^n$ is again equipped with the \emph{Zariski topology}, where a subset is closed if it is the set of common zeros of a set of \emph{homogeneous} polynomials in  $\F[X_1,\dots,X_{n+1}]$. We use $(n+1)$-tuples $(x_1,\dots,x_{n+1})$ to represent points in $\proj^n$, called \emph{homogeneous coordinates}.

For a vector space $V$ over $\F$ of dimension $n+1$, where $n\in\N$, define the projective space $\proj V=(V\setminus\{\mathbf{0}\})/\sim$, where $\sim$ is again the equivalence relation defined by scaling. By fixing a coordinate system of $V$ and identifying it with $\aff^{n+1}$, we may identify $\proj V$ with $\proj^n$.

%Irreducible decomposition of an affine variety mirrors the factoring of an ideal into primary ideals.  Finally, note that the affine points are in 1-1 correspondence with {\em maximal} ideals; it is a simple reformulation of Hilbert's Nullstellensatz.

\paragraph{Varieties.}
\emph{Varieties} in this paper refer to either projective or affine varieties.
A \emph{projective (resp. affine) variety} is simply a  closed subset of a projective (resp. affine) subspace.
%\footnote{The empty set is often  regarded as a variety, but we exclude it in this paper for convenience. Also, we do \emph{not} assume varieties are irreducible. The notion of varieties in this paper are often called \emph{(nonempty) algebraic sets} in literature.}
 If $\mathcal{V}_1$ and $\mathcal{V}_2$ are closed subsets of a projective or affine space and $\mathcal{V}_1\subseteq \mathcal{V}_2$, we say $\mathcal{V}_1$ is a \emph{(closed) subvariety} of $\mathcal{V}_2$.

A variety is {\em reducible} if it is the union of finitely many proper subvarieties, and otherwise \emph{irreducible}.
Affine and projective spaces are irreducible.
% (In this paper, varieties are allowed to be reducible.)
A variety $\mathcal{V}$ can be uniquely written as the union of finitely many maximal irreducible subvarieties, which are called the \emph{irreducible components} of $\mathcal{V}$. 

A projective or affine variety is called a \emph{hypersurface} (resp. \emph{hyperplane}) if it is definable by a single polynomial (resp. single  linear polynomial). 

\paragraph{Hilbert's Nullstellensatz.}
An ideal $I$ of a commutative ring $R$ is \emph{radical} if $a^m\in I$ implies $a\in I$ for every $a\in R$ and $m\in\N^+$.
For an ideal $I$ of $\F[X_1,\dots,X_n]$, denote by $\mathcal{V}(I)$ the subvariety of $\aff^n$ defined by the polynomials in $I$. 
Define $\mathcal{V}(f_1,\dots, f_k)=\mathcal{V}(\langle f_1,\dots, f_k\rangle)$ for $f_1,\dots,f_k\in \F[X_1,\dots,X_n]$.
For a subvariety $\mathcal{V}$ of $\aff^n$, denote by $I(\mathcal{V})$ the ideal of $\F[X_1,\dots,X_n]$ consisting of all the polynomials vanishing on $\mathcal{V}$. 
\emph{Hilbert's Nullstellensatz} states that the map $\mathcal{V}\mapsto I(\mathcal{V})$ is an inclusion-reversing one-to-one correspondence between the subvarieties of $\aff^n$ and the radical ideals of $\F[X_1,\dots,X_n]$, with the inverse map $I\mapsto \mathcal{V}(I)$. 

For a subvariety $\mathcal{V}$ of $\aff^n$, define $\F[\mathcal{V}]:=\F[X_1,\dots,X_n]/I(\mathcal{V})$, called the \emph{coordinate ring} of $\mathcal{V}$.

\paragraph{Projective Nullstellensatz.}
Consider the polynomial ring $R=\F[X_1,\dots,X_{n+1}]$. It can be written as a direct sum $R=\bigoplus_{d=0}^\infty R_d$ where each $R_d$ denotes the  space of degree-$d$ homogeneous polynomials, called the \emph{homogeneous part of degree $d$} of $R$ or simply the \emph{degree-$d$ part} of $R$. For an ideal $I$ of $R$ and $d\in\N$, let $I_d:=I\cap R_d$, called the \emph{degree-$d$ part} of $I$. We say $I$ is a \emph{homogeneous ideal} if $I=\bigoplus_{d=0}^\infty I_d$.
For a homogeneous ideal $I$ of $R$, we have $R/I=\bigoplus_{d=0}^\infty (R/I)_d$ where $(R/I)_d:=R_d/I_d$.

For a homogeneous ideal $I$ of $R$, denote by $\mathcal{V}(I)$ the subvariety of $\proj^n$ defined by the homogeneous polynomials in $I$. 
Define $\mathcal{V}(f_1,\dots, f_k)=\mathcal{V}(\langle f_1,\dots, f_k\rangle)$ for homogeneous polynomials $f_1,\dots,f_k\in R$.
For a subvariety $\mathcal{V}$ of $\proj^n$, denote by $I(\mathcal{V})$ the ideal generated by the homogeneous polynomials vanishing on $\mathcal{V}$, which is a homogeneous ideal.
The \emph{projective Nullstellensatz} states that the map $\mathcal{V}\mapsto I(\mathcal{V})$ is an  inclusion-reversing one-to-one correspondence between the nonempty subvarieties of $\proj^n$ and the radical homogeneous ideals of $R$ properly contained in $\langle X_1,\dots,X_{n+1}\rangle$, with the inverse map $I\mapsto \mathcal{V}(I)$. 

For a subvariety $\mathcal{V}\subseteq\proj^n$ and the corresponding homogeneous ideal $I=I(\mathcal{V})$, we say $R/I$ is the \emph{homogeneous coordinate ring} of $\mathcal{V}$.

\paragraph{Morphisms.} Let $\mathcal{V}_1\subseteq \aff^n$ and $\mathcal{V}_2\subseteq\aff^m$ be affine varieties. A \emph{morphism}  from $\mathcal{V}_1$ to $\mathcal{V}_2$ is a map $f:\mathcal{V}_1\to \mathcal{V}_2$ that is a restriction of a polynomial map $\aff^n\to\aff^m$. Such a morphism $f$ is associated with a ring homomorphism $f^\sharp: \F[\mathcal{V}_2]\to \F[\mathcal{V}_1]$, making $\F[\mathcal{V}_1]$ an algebra over $\F[\mathcal{V}_2]$.
We say $f$ is \emph{finite}  if $\F[\mathcal{V}_1]$ is finitely generated as an $\F[\mathcal{V}_2]$-module.

Let $f: \mathcal{V}_1\to \mathcal{V}_2$ be a map between projective varieties $\mathcal{V}_1$ and $\mathcal{V}_2$. We say $f$ is a morphism from $\mathcal{V}_1$ to $\mathcal{V}_2$
if there exists a collection of open subsets $\{U_i\}_{i\in I}$ of $\mathcal{V}_2$ such that $\mathcal{V}_2=\bigcup_{i\in I} U_i$ (i.e., $\{U_i\}_{i\in I}$ is an open cover of $\mathcal{V}_2$) and for each $i\in I$, the restriction $f|_{f^{-1}(U_i)}: f^{-1}(U_i)\to U_i$ is a morphism between affine varieties. Furthermore,  if  each $f|_{f^{-1}(U_i)}$ is finite, then we say $f$ is finite. Finiteness does not depend on the choice of the affine open cover. Namely, if $f: \mathcal{V}_1\to \mathcal{V}_2$ is a finite morphism between projective varieties $\mathcal{V}_1$ and $\mathcal{V}_2$, and $U$ is an open subset of $\mathcal{V}_2$ such that  $f|_{f^{-1}(U)}: f^{-1}(U)\to U$ is a morphism between affine varieties, then $f|_{f^{-1}(U)}$ is also finite.

The image of a morphism $f:\mathcal{V}_1\to\mathcal{V}_2$ is denoted by $\mathrm{Im}(f)$ or $f(\mathcal{V}_1)$.
The image of a closed set under a finite morphism is still closed. The composition of two finite morphisms is still finite.

%All maps between varieties in this paper are morphisms.
%A morphism $f: V\to W$ is called {\em dominant} if $\overline{\im(f)}=W$.

%\smallskip
%The affine space $\aff^d$ may be regarded as a subset of $\proj^d$ via the map $(x_1,\dots,x_d)\mapsto (1,x_1,\dots,x_d)$. 
%Then the subspace topology of $\aff^d$ induced from the Zariski topology of $\proj^d$ is just the Zariski topology of $\aff^d$.
%The set $\proj^d\setminus \aff^d$ is the projective subspace of $\proj^d$ defined by $X_0=0$, called the {\em hyperplane at infinity}.
 
%For an algebraic subset $V$ of $\aff^d\subseteq\proj^d$, the smallest algebraic subset $V'$ of $\proj^d$ containing $V$ (i.e. the intersection of all algebraic subsets containing $V$) is the {\em projective closure} of $V$, and we have $V'\cap \aff^d=V$. To see this, note that for $P=(x_1,\dots,x_d)\in \aff^d\setminus V$, there exists a polynomial $Q\in\aff[X_1,\dots,X_d]$ of degree $D\in\N$ not vanishing on $P$ (but vanishing on $V$). Then its homogenization $Q'\in\aff[X_0,\dots,X_d]$, defined by replacing each monomial $M=\prod_{i=1}^d X_i^{d_i}$ by  $X_0^{D-\deg(M)}\prod_{i=1}^d X_i^{d_i}$, does not vanish on $(1,x_1,\dots,x_d)$. So, $(1,\mathbf x)\notin V'$.

\paragraph{Dimension.}
The {\em dimension} of an irreducible variety $\mathcal{V}$, denoted by  $\dim(\mathcal{V})$, is  the largest integer $m$ such that there exists a chain of irreducible varieties $\emptyset\subsetneq \mathcal{V}_0\subsetneq  \mathcal{V}_1\subsetneq\cdots\subsetneq  \mathcal{V}_m=\mathcal{V}$. More generally, the dimension of a nonempty variety is the maximal dimension of its irreducible components. 
We define the dimension of an empty set to be $-\infty$.
A variety is \emph{equidimensional} if its irreducible components have the same dimension.

If $\pi: \mathcal{V}\to\mathcal{V}'$ is a finite morphism, then $\dim(\mathcal{V})=\dim(\pi(\mathcal{V}))$.
%The dimension of  the  empty set is $-1$ by convention. One dimensional varieties are called {\em curves}.

\paragraph{Degree.}
The {\em degree} of an irreducible subvariety $\mathcal{V}$ of $\proj^n$ (resp. $\aff^n$), denoted by $\deg(\mathcal{V})$,  is the number of intersections of $\mathcal{V}$ with a projective (resp. affine) subspace of codimension $\dim(\mathcal{V})$ in general position. 
More generally, we define the degree of a subvariety of $\proj^n$ or $\aff^n$  to be the sum of the degrees of its irreducible components.

%Suppose  $\mathcal{V}=\mathcal{V}(f_1,\dots,f_k)$ is a projective or affine variety defined by nonzero polynomials $f_1,\dots,f_k$, where the degree of $f_i$ is $d_i$ for $i\in [k]$ and $d_1\geq d_2\geq \dots\geq d_k$. If the codimension of $\mathcal{V}$ is $k'\leq k$, then \emph{B\'ezout's inequality} (see, e.g., \cite{Hei83}) gives
%\[
%\deg(\mathcal{V})\leq \prod_{i=1}^{k'} d_i.
%\]

\paragraph{Projective closure.}
The affine $n$-space $\aff^n$ may be regarded as an open subset of $\proj^n$ via the map $(x_1,\dots,x_n)\mapsto (x_1,\dots,x_n, 1)$. 
The complement $H_\infty:=\proj^n\setminus \aff^n$ is a hyperplane of $\proj^n$ defined by $X_{n+1}=0$, called the \emph{hyperplane at infinity}.
 For an affine subvariety $\mathcal{V}$ of $\aff^n\subseteq\proj^n$, the smallest projective subvariety of $\proj^n$ containing $\mathcal{V}$ is the \emph{projective closure} of $\mathcal{V}$, which we denote by $\mathcal{V}_\mathrm{cl}$.
It is known that $\mathcal{V}_\mathrm{cl}\cap \aff^n=\mathcal{V}$, $\dim(\mathcal{V}_\mathrm{cl})=\dim(\mathcal{V})$, and $\deg(\mathcal{V}_\mathrm{cl})=\deg(\mathcal{V})$. 

\paragraph{Joins of disjoint projective varieties.}

For two distinct points $p,q\in\proj^n$, denote by $\overline{pq}$ the unique projective line passing through them.
For two \emph{disjoint} projective subvarieties $\mathcal{V}_1,\mathcal{V}_2\subseteq\proj^n$, define the \emph{join} $J(\mathcal{V}_1, \mathcal{V}_2)$ of $\mathcal{V}_1$ and $\mathcal{V}_2$ as
\[
J(\mathcal{V}_1, \mathcal{V}_2):=\bigcup_{p\in\mathcal{V}_1,q\in\mathcal{V}_2} \overline{pq}.
\]

\begin{lemma}[{\cite[Examples~6.17, 11.36, and 18.17]{Har13}}]\label{lem_join}
$J(\mathcal{V}_1, \mathcal{V}_2)$ is a subvariety of $\proj^n$ of dimension $\dim(\mathcal{V}_1)+\dim(\mathcal{V}_2)+1$ and degree at most $\deg(\mathcal{V}_1)\cdot \deg(\mathcal{V}_2)$. 
\end{lemma}

We also need the following facts.

\begin{lemma}[{\cite[Exercise~11.6 and Corollary~18.5]{Har13}}]\label{lem_hyperplane}
Let $\mathcal{V}$ be a nonempty equidimensional subvariety of $\proj^n$ and $H$ a hypersurface of $\proj^n$ not containing an irreducible component of $\mathcal{V}$. Then  $\mathcal{V}\cap H$ is an equidimensional subvariety of dimension $\dim(\mathcal{V})-1$ and  degree at most $\deg(\mathcal{V})\cdot \deg(H)$ (or an empty set if $\dim(\mathcal{V})=0$).
\end{lemma}
 
 \begin{corollary}\label{cor_disjoint}
Let $\mathcal{V}$ be a subvariety of $\proj^n$ of dimension $r$, where $0\leq r<n$.  Then there exists an $(n-r-1)$-subspace $W$ disjoint from $\mathcal{V}$.
\end{corollary}
\begin{proof}
It suffices to show that there exist hyperplanes $H_1,\dots,H_{r+1}$, such that $\mathcal{V}_i:=\mathcal{V}\cap (\bigcap_{j=1}^{i} H_j)$ is empty for some $i\leq r+1$.
We may inductively choose each $H_i$ such that $C\not\subseteq H_i$ for every irreducible component $C$ of $\mathcal{V}_{i-1}$, so that $\dim(\mathcal{V}_i)\leq \dim(\mathcal{V}_{i-1})-1$ by \cref{lem_hyperplane}.
So $\mathcal{V}_i=\emptyset$ for some $i\leq r+1$.
\end{proof}

 \begin{lemma}[{\cite[Section~I.6.2, Theorem~6]{Sha13}}]\label{lem_dimcap}
Suppose $\mathcal{V}_1$ and $\mathcal{V}_2$ are subvarieties of $\proj^n$ and $\dim(\mathcal{V}_1)+\dim(\mathcal{V}_1)\geq n$. Then $\mathcal{V}_1\cap \mathcal{V}_2\neq\emptyset$ and $\dim(\mathcal{V}_1\cap\mathcal{V}_2)\geq \dim(\mathcal{V}_1)+\dim(\mathcal{V}_1)-n$.
\end{lemma}

%Let $V\subseteq \aff^n$, $W\subseteq\aff^m$ be affine varieties. A {\em morphism} from $V$ to $W$ is a function $f:V\to W$ that is a restriction of a polynomial map $\aff^n\to\aff^m$.
%A morphism $f: V\to W$ is called {\em dominant} if $\overline{\im(f)}=W$.
%The preimage of a closed subset under a morphism is closed (i.e. morphisms are {\em continuous} in the Zariski topology).
%
%For a polynomial map $f:\aff^n\to\aff^m$ and an affine variety $V\subseteq \aff^n$, $W:=\overline{f(V)}$ is also an affine variety (i.e., it is irreducible). To see this, assume to the contrary that $W$ is the union of two proper closed subsets $W_1$ and $W_2$. By the definition of closure, $f(V)$ is not contained in either $W_1$ or $W_2$, i.e., it intersects both. Then $f^{-1}(W_1)\cap V$ and $f^{-1}(W_2)\cap V$ are two proper closed subsets of $V$, and their union is $V$. This contradicts the irreducibility of $V$.

\section{Proof of the Main Theorem}\label{sec_main}

In this section, we prove the Main Theorem (\cref{thm_main})  together with \cref{thm_dnoetherproj} and \cref{thm_dnoetheraff}. In \cref{sec_equi}, we show that it suffices to consider equidimensional or irreducible subvarieties of dimension $n-k-1$. \cref{sec_chow} contains an introduction to Chow forms. In \cref{sec_construction}, we present the explicit constructions and complete the proof of \cref{thm_main}.
As a product, \cref{thm_dnoetherproj} and \cref{thm_dnoetheraff} are also proved in \cref{sec_construction}.

\subsection{Reducing to the Case of Equidimensional or Irreducible Varieties}\label{sec_equi}

The following lemma states that to construct  $k$-subspace families that are evasive for  subvarieties of $\proj^n$, it suffices to consider equidimensional subvarieties of dimension $n-k-1$ (i.e., codimension $k+1$).

\begin{lemma}\label{lem_equi}
Let $n,d\in\N^+$ and $k\in\{0,1,\dots,n-1\}$. 
Let $\mathcal{F}$ be the family of all equidimensional subvarieties of $\proj^n$ of dimension $n-k-1$ and degree at most $d$.
Then an $(\mathcal{F}, \epsilon)$-evasive $k$-subspace family is also  $(n,d,\epsilon)$-evasive.
\end{lemma}

The proof of \cref{lem_equi} is based on the following claim. 

\begin{claim}\label{claim_replace}
Let  $\mathcal{V}$ be an irreducible subvariety of $\proj^n$.
There exists a subvariety $\widetilde{\mathcal{V}}\subseteq \proj^n$ of dimension $n-k-1$ and degree at most $\deg(\mathcal{V})$ such that any $k$-subspace of $\proj^n$ that evades $\widetilde{\mathcal{V}}$ also evades $\mathcal{V}$. 
\end{claim}
\begin{proof}
If $\dim(\mathcal{V})=n-k-1$, then just let $\widetilde{\mathcal{V}}=\mathcal{V}$.

Now assume $\dim(\mathcal{V})<n-k-1$. 
Let $t=(n-k-1)-\dim(\mathcal{V})-1$ and let $\widetilde{\mathcal{V}}$ be the join of $\mathcal{V}$ and a $t$-subspace disjoint from $\mathcal{V}$ (which exists by \cref{cor_disjoint}).
Then  $\widetilde{\mathcal{V}}$ is a projective subvariety of dimension $n-k-1$ and degree at most $\deg(\mathcal{V})$ by \cref{lem_join}.  Suppose $W$ is a $k$-subspace that evades $\widetilde{\mathcal{V}}$.
Then $W$ is disjoint from $\widetilde{\mathcal{V}}\supseteq \mathcal{V}$.
So $W$ also evades $\mathcal{V}$.

Finally, assume $\dim(\mathcal{V})>n-k-1$. Let $t=\dim(\mathcal{V})-(n-k-1)$. 
By \cref{lem_hyperplane}, there exist $t$ hyperplanes $H_1,\dots, H_t$ of $\proj^n$ such that $\mathcal{V}\cap \bigcap_{i=1}^t H_i$ is equidimensional of dimension $n-k-1$ and degree at most $\deg(\mathcal{V})$. Let $\widetilde{\mathcal{V}}=\mathcal{V}\cap \bigcap_{i=1}^t H_i$.   Suppose $W$ is a $k$-subspace that evades $\widetilde{\mathcal{V}}$.
Then $W\cap \widetilde{\mathcal{V}}=(W\cap \mathcal{V})\cap \bigcap_{i=1}^t H_i=\emptyset$.
Again by \cref{lem_hyperplane}, we have $\dim(W\cap \mathcal{V})\leq t-1=\dim(\mathcal{V})+\dim(W)-n$.
%By \cref{lem_dimcap}, the previous inequality is an equality.
So $W$ also evades $\mathcal{V}$.
\end{proof}

\begin{proof}[Proof of \cref{lem_equi}]
%Let $\mathcal{H}$ be an $(\mathcal{F}, \epsilon)$-evasive $k$-subspace family.
Consider a projective subvariety $\mathcal{V}\subseteq\proj^n$ of degree at most $d$.
Let $\mathcal{V}_1,\dots,\mathcal{V}_s$ be the irreducible components of $\mathcal{V}$.
For each $i\in [s]$, use \cref{claim_replace} to choose a projective subvariety  $\widetilde{\mathcal{V}}_i\subseteq\proj^n$ of  dimension $n-k-1$ and degree at most $\deg(\mathcal{V}_i)$ such that any $k$-subspace that evades $\widetilde{\mathcal{V}}_i$ also evades $\mathcal{V}_i$.
Let $\widetilde{\mathcal{V}}=\bigcup_{i=1}^s \widetilde{\mathcal{V}}_i$.  
Then  $\widetilde{\mathcal{V}}\in\mathcal{F}$. By construction, any $k$-subspace that evades $\widetilde{\mathcal{V}}$ also evades $\mathcal{V}$. It follows that an $(\mathcal{F}, \epsilon)$-evasive $k$-subspace family is also $(n,d,\epsilon)$-evasive.
\end{proof}

We further reduce to the case of irreducible varieties at the cost of blowing up the parameter $\epsilon$ by a factor of $d$. This is useful as we need irreducibility later in \cref{lem_span}.

\begin{lemma}\label{lem_irred}
Let $n,d\in\N^+$ and $k\in\{0,1,\dots,n-1\}$. 
Let $\mathcal{F}'$ be the family of all irreducible subvarieties of $\proj^n$ of dimension $n-k-1$ and degree at most $d$.
Then an $(\mathcal{F}', \epsilon)$-evasive $k$-subspace family is also an $(n,d, d\epsilon)$-evasive  $k$-subspace family.
\end{lemma}

\begin{proof}
Let $\mathcal{F}$ be as in \cref{lem_equi}. Each $\mathcal{V}\in\mathcal{F}$ has at most $d$ irreducible components, which are all in $\mathcal{F}'$ since their degrees are bounded by $d$. By definition and the union bound, if a $k$-subspace family $\mathcal{H}$ is  $(\mathcal{F}', \epsilon)$-evasive, then it is also $(\mathcal{F}, d\epsilon)$-evasive.  Combining this with \cref{lem_equi} proves the lemma. 
\end{proof}

\subsection{Chow Forms}\label{sec_chow}

By \cref{lem_equi} and \cref{lem_irred}, we only need to evade equidimensional or irreducible projective subvarieties of codimension $k+1$. %It is well known that a $k$-subspace in general position is disjoint from these varieties.
The ``bad'' $k$-subspaces that intersect such a variety $\mathcal{V}$ form a hypersurface of the Grassmannian defined by a single form called the \emph{Chow form} of $\mathcal{V}$. We now explain the basic theory of Chow forms.

\paragraph{Grassmannians.}

Let $n\in\N$ and $k\in \{0,1,\dots, n-1\}$. The {\em Grassmannian} $\gr(k+1, n+1)$ is the set of all $(k+1)$-dimensional linear subspaces of $\aff^{n+1}$.  By taking the quotient modulo scalars, it may also be identified with the set of all  $k$-subspaces of $\proj^n$, which we denote by $\pgr(k, n)$.

\paragraph{The Pl\"ucker embedding and Pl\"ucker coordinates.}
Consider a linear subspace $W\in \gr(k+1, n+1)$. 
The simplest way of representing $W$ is using a $(k+1)\times (n+1)$ matrix $A$ over $\F$ such that $W$ equals the row space of $A$. We call such a matrix $A$ a \emph{generating matrix} of $W$. 
For convenience, we also say $A$ is a generating matrix of $\proj W\in \pgr(k,n)$.

The entries of $A$ are called the \emph{(primal) Stiefel coordinates} of $W$.
However, note that  $A$ is not uniquely determined by $W$ since for any $(k+1)\times (k+1)$ invertible matrix $M$ over $\F$, the matrix $MA$ is also a generating matrix of $W$.

Another way of representing $W$ is using the vector $(\det A_{[k+1], S})_{S\in \binom{[n+1]}{k+1}}$  of  maximal minors of a generating matrix $A$ of $W$. For a $(k+1)\times (k+1)$  invertible matrix $M$ over $\F$,  replacing $A$ by $MA$ corresponds to multiplying all the maximal minors $\det A_{[k+1], S}$ by $\det M\in\F^\times$. To remove ambiguity, we could view $(\det A_{[k+1], S})_{S\in \binom{[n+1]}{k+1}}$ as a point in the projective space $\proj^{\binom{n+1}{k+1}-1}$,  which is then uniquely determined by $W$.
This leads to the definition of the \emph{Pl\"ucker embedding}.

\begin{definition}[Pl\"ucker embedding]
Define $\phi: \gr(k+1,n+1)\to \proj^{\binom{n+1}{k+1}-1}$
by 
\[
\phi(W)=(\det A_{[k+1], S})_{S\in \binom{[n+1]}{k+1}}   
\]
where $A$ is a generating matrix of $W$.
\end{definition}

The Pl\"ucker embedding embeds the Grassmannian  $\gr(k+1,n+1)$  in $\proj^{\binom{n+1}{k+1}-1}$ as an irreducible projective subvariety, as stated by the following theorem. See, e.g., \cite{Har13, Ful97} for proofs.

\begin{theorem}
The Pl\"ucker embedding $\phi$ is a well-defined injective map whose image is an irreducible projective subvariety of $\proj^{\binom{n+1}{k+1}-1}$.
\end{theorem}
 
The homogeneous coordinates $(\det A_{[k+1], S})_{S\in \binom{[n+1]}{k+1}}$ of $\phi(W)$ are called the \emph{(primal) Pl\"ucker coordinates} of $W$. 

Denote by $R:=\F\left[X_S: S\in \binom{[n+1]}{k+1}\right]$ the homogeneous coordinate ring of $\proj^{\binom{n+1}{k+1}-1}$.
The irreducible projective subvariety $\phi(\gr(k+1,n+1))$ is defined by a homogeneous prime ideal of $R$, which we denoted by $I$. Then $R/I$ is the homogeneous coordinate ring of $\phi(\gr(k+1,n+1))$.
The ideal $I$ contains precisely the polynomial relations that the Pl\"ucker coordinates need to satisfy.
It is also known that $I$ is generated by certain quadratic forms, known as the \emph{Pl\"ucker relations}. See \cite{Har13, Ful97} for details.

\paragraph{Dual Pl\"ucker coordinates.}

Alternatively, we could  represent a linear subspace $W\in \gr(k+1, n+1)$  by  an $(n-k)\times (n+1)$ matrix $B$ over $\F$ whose rows specify the linear equations defining $W$. We call such a matrix $B$ a \emph{parity check matrix} of $W$. 
For convenience, we also say $B$ is a parity check matrix of $\proj W\in \pgr(k,n)$.

The entries of $B$ are called the \emph{dual Stiefel coordinates} of $W$.
This gives another  embedding $\phi^\vee: \gr(k+1,n+1)\to \proj^{\binom{n+1}{n-k}-1}=\proj^{\binom{n+1}{k+1}-1}$, defined by 
\[
\phi^\vee(W)=(\det B_{[n-k], S})_{S\in \binom{[n+1]}{n-k}}.
\]
The homogeneous coordinates $(\det B_{[n-k], S})_{S\in \binom{[n+1]}{n-k}}$ of $\phi^\vee(W)$ are called the \emph{dual Pl\"ucker coordinates} of $W$.\footnote{Many authors use ``primal'' and ``dual''  in the opposite way (e.g., \cite{DS95}).} 
In fact, it is known that dual Pl\"ucker coordinates are equivalent to primal  Pl\"ucker coordinates. Namely, if $W\in \gr(k+1, n+1)$ has primal Pl\"ucker coordinates $(c_S)_{S\in \binom{[n+1]}{k+1}}$, then it has dual Pl\"ucker coordinates $(c'_S)_{S\in \binom{[n+1]}{n-k}}$ with $c'_S=(-1)^{\sum_{i\in S} i-\sum_{i\in [k+1]} i} \cdot c_{[n+1]\setminus S}$ (see, e.g., \cite{JT13}). 
%So the map $\phi^\vee$ is the same as $\phi$ modulo sign changes and a permutation of coordinates.

%Let  $B^*$ be the $(n-k)\times n$ variable matrix whose $(i,j)$-th entry is $Y_{i,j}$. 
%Define the ring homomorphism $\hat{\phi}^*: R\to \F[\{Y_{i,j}\}_{i\in [n-k], j\in [n]}]$
%that substitutes each variable $X_S$ by $\det B^*_{[n-k], [n]\setminus S}$. 
%That is, $\hat{\phi}^*$ substitute a dual Pl\"ucker coordinate with the determinant of the corresponding matrix of dual Stiefel coordinates.

%\paragraph{The homogeneous coordinate ring of the Grassmannian.}
%When $k>1$, the image $\phi(\gr(k,n))$ is not the whole space $\proj^{\binom{n}{k}-1}$ and is defined by some polynomials, which express the relations that the Pl\"ucker coordinates always satisfy.
%Formally, let $R=\bigoplus_{d\in\N} R_d$ be the homogeneous coordinate ring of $\proj^{\binom{n}{k}-1}$, where $R_d$ consists of the homogeneous polynomials of degree $d$ in the variables $X_S$, $S\in \binom{[n]}{k}$. 
%Let $A^*$ be a $k\times n$ variable matrix whose $(i,j)$-th entry is a variable $Y_{i,j}$.
%Define the ring homomorphism
%\[
%\phi^\sharp: R\to\F[Y_{i,j}: i\in [k], j\in [n]]
%\]
%that substitutes each variable $X_S$ by $\det A^*_{[k], S}$. 
%%That is, $\phi^\sharp$ sends a polynomial $P\left((X_S)_{S\in \binom{[n]}{k}}\right)$ to $P\left((\det A_{[k], S})_{S\in \binom{[n]}{k}}\right)$.
%Define $I=\overline{\K}er(\phi^\sharp)$. Then $I$ is precisely the set of polynomials of $R$ vanishing on $\phi(\gr(k,n))$. It is a homogeneous prime ideal of $R$.

\paragraph{Chow forms.}
Recall that we denote by $\pgr(k, n)$ the set of all  $k$-subspaces of $\proj^n$.
By identifying $\gr(k+1, n+1)$ with $\pgr(k, n)$ via $W\mapsto \proj W$, we regard $\phi$ and $\phi^\vee$ as maps from $\pgr(k, n)$ to $\proj^{\binom{n+1}{k+1}-1}$.

 We also need the notion of \emph{associated hypersurfaces}.
\begin{definition}[Associated hypersurface \cite{GKZ94}]
For an irreducible subvariety  $\mathcal{V}\subseteq \proj^{n}$ of dimension $n-k-1$, define the \emph{associated hypersurfaces} $\mathcal{Z}_\mathcal{V}$ of $\mathcal{V}$ to be the set of $k$-subspaces intersecting $\mathcal{V}$, i.e.,
\[
 \mathcal{Z}_\mathcal{V}:=\{W\in \pgr(k, n): \mathcal{V}\cap  W\neq\emptyset\}.
\]
\end{definition}

The term ``associated hypersurface'' is justified by the following theorem.

\begin{theorem}\label{thm_chow}
Let $\mathcal{V}\subseteq \proj^{n}$ be an irreducible projective subvariety of dimension $n-k-1$ and degree $d\in\N^+$.
Then there exists a nonzero homogeneous polynomial $P_\mathcal{V}\in R=\F\left[X_S: S\in \binom{[n+1]}{k+1}\right]$ of degree $d$ such that  $\phi(\mathcal{Z}_\mathcal{V})$ is defined by $P_\mathcal{V}$ as a subvariety of $\phi(\pgr(k, n))$. That is,
\[
\phi(\mathcal{Z}_\mathcal{V})=\phi(\pgr(k, n)) \cap \mathcal{V}(P_\mathcal{V}).
\]
Moreover, $\mathcal{R}_\mathcal{V}:=P_\mathcal{V}+I\in (R/I)_d$ is uniquely determined by $\mathcal{V}$ up to scalars.
\end{theorem}
\cref{thm_chow} is explicitly stated as \cite[Theorem~1.1 and Corollary~2.1]{DS95}.
A proof can be found in \cite[Section~3.2]{GKZ94}. 
We briefly explain how to find a polynomial $P_\mathcal{V}$ satisfying \cref{thm_chow}:
Firstly, it can be shown using the trick of  dimension counting via incidence varieties that $\phi(\mathcal{Z}_\mathcal{V})$ is an irreducible projective subvariety of the Grassmannian $\phi(\pgr(k, n))$ of codimension one \cite[Section~3.2, Proposition~2.2]{GKZ94}. Secondly, the homogeneous coordinate ring $R/I$ of the Grassmannian is known to be a \emph{unique factorization domain} \cite[Chapter~9]{Ful97}. These two facts imply that the homogeneous ideal of $R/I$ defining $\phi(\mathcal{Z}_\mathcal{V})$ is a \emph{principal ideal}. Choose $\mathcal{R}_\mathcal{V}$ to be a generator of this principal ideal, which is unique up to scalars. Then lift $\mathcal{R}_\mathcal{V}\in R/I$ to $P_\mathcal{V}\in R$.

Now we are ready to define the Chow form of projective subvarieties.

\begin{definition}[Chow form]
Let $\mathcal{V}\subseteq \proj^{n}$ be an irreducible  subvariety of dimension $n-k-1$ and degree $d\in\N^+$.
Define the \emph{Chow form of $\mathcal{V}$ in  Pl\"ucker coordinates}, or simply the \emph{Chow form} of $\mathcal{V}$, to be $\mathcal{R}_\mathcal{V}\in (R/I)_d$ as in \cref{thm_chow}.

More generally, for an equidimensional  subvariety $\mathcal{V}=\bigcup_{i=1}^s \mathcal{V}_i\subseteq\proj^{n}$ of dimension $n-k-1$ and degree $d$, where $\mathcal{V}_1,\dots,\mathcal{V}_s$ are the irreducible components of $\mathcal{V}$, the \emph{Chow form} of $\mathcal{V}$ is $\mathcal{R}_\mathcal{V}:=\prod_{i=1}^s \mathcal{R}_{\mathcal{V}_i}\in (R/I)_d$. It is uniquely determined by $\mathcal{V}$ up to scalars. 
\end{definition}

As a $k$-subspace intersects $\mathcal{V}=\bigcup_{i=1}^s \mathcal{V}_i$ iff it intersects some $\mathcal{V}_i$, we see from \cref{thm_chow} that the Chow form $\mathcal{R}_{\mathcal{V}}$ of an equidimensional projective subvariety $\mathcal{V}$ of dimension $n-k-1$  vanishes precisely at the set of $k$-subspaces that intersect $\mathcal{V}$.

\begin{exmp}\label{exmp_hyp}
Let $k=0$. Let $\mathcal{V}\subseteq\proj^{n}$ be a hypersurface defined by a nonzero homogeneous polynomial $P\in \F[X_1,\dots,X_{n+1}]=R$.
The ideal $I$ of $R$ is zero in this case.
And the Chow form $\mathcal{R}_{\mathcal{V}}$ of $\mathcal{V}$  is simply $P$ (up to a scalar).
\end{exmp}

\begin{exmp}\label{exmp_gr}
Let $V\in \gr(n-k, n+1)$ and $W\in \gr(k+1, n+1)$. Choose   matrices $A, B\in \F^{(k+1)\times (n+1)}$ such that $A$ is a  generating matrix of $W$ and  $B$ is a parity check matrix of $V$.
Then $\proj V\cap \proj W\neq \emptyset$ iff $\dim (V\cap W)>0$, which holds iff
$\det(A B^T)=0$.
On the other hand, we have
\[
\det(A B^T)=\sum_{S\in \binom{[n+1]}{k+1}} \det(A_{[k+1], S}) \cdot \det((B^T)_{S, [k+1]})=\sum_{S\in \binom{[n+1]}{k+1}} \det(A_{[k+1], S}) \cdot \det(B_{[k+1], S}),
\]
where the first equation is known as the \emph{Cauchy--Binet formula} (see, e.g., \cite{FSS14}).
So $P_{\proj V}\in R_1$ is a linear polynomial whose coefficients are given by the dual Pl\"ucker coordinates $(\det B_{[k+1], S})_{S\in \binom{[n+1]}{k+1}}$ of $V$ (up to a scalar). The degree-one part $I_1$ of $I$ is zero as $I$ is generated by quadratic forms. So the Chow form $\mathcal{R}_{\proj V}\in (R/I)_1=R_1$ is simply $P_{\proj V}$.
\end{exmp}

\paragraph{Chow forms in Stiefel coordinates.}

We may also express the Chow form in Stiefel coordinates, i.e., in the entries of a generating matrix of a linear subspace.
This expression has the advantage that it is an actual polynomial rather than a member of the abstract vector space $(R/I)_d$.

Formally, let $A^*$ be a $(k+1)\times (n+1)$ variable matrix whose $(i,j)$-th entry is a variable $Y_{i,j}$.
Define the ring homomorphism
\[
\phi^\sharp: R=\F\left[X_S: S\in \binom{[n+1]}{k+1}\right] \to\F[Y_{i,j}: i\in [k+1], j\in [n+1]]
\]
that sends each variable $X_S$ to $\det(A^*_{[k+1], S})$. Define the \emph{Chow form of $\mathcal{V}$ in Stiefel coordinates} to be 
\[
\widetilde{\mathcal{R}}_\mathcal{V}:=\phi^\sharp(P_\mathcal{V}) \in \F[Y_{i,j}: i\in [k+1], j\in [n+1]]
\]
where $P_\mathcal{V}\in R_d$ is a lift of $\mathcal{R}_\mathcal{V}\in (R/I)_d$.
Note that $I$ is precisely the kernel of $\phi^\sharp$. 
So $\widetilde{\mathcal{R}}_\mathcal{V}$ is uniquely determined by $\mathcal{V}$ up to scalars.
By construction, for any $W\in \gr(k+1,n+1)$ and generating matrix $A=(a_{i,j})_{i\in [k+1], j\in [n+1]}$ of $W$, we have $P_\mathcal{V}(\phi(W))=\widetilde{\mathcal{R}}_\mathcal{V}(A):=\widetilde{\mathcal{R}}_\mathcal{V}(a_{1,1},\dots,a_{k+1,n+1})$.
So $\widetilde{\mathcal{R}}_\mathcal{V}$ vanishes at $A$ iff $\proj W\in \pgr(k, n)$ intersects $\mathcal{V}$.

\paragraph{Chow forms in dual  Stiefel coordinates.}
Similarly, we may express the Chow form in dual Stiefel coordinates, i.e., in the entries of a parity check matrix of a linear subspace. 

More specifically, choose a homogeneous polynomial $Q_\mathcal{V}\in \F\left[X_S: S\in \binom{[n+1]}{n-k}\right]$ that defines the set of $k$-subspaces intersecting $\mathcal{V}$ in terms of dual Pl\"ucker coordinates.  As primal and dual  Pl\"ucker coordinates are equivalent,   $Q_\mathcal{V}$ can be obtained from the polynomial $P_\mathcal{V}$ above by simply negating and renaming variables. Next, compose  $Q_\mathcal{V}$ with a ring homomorphism  that substitutes  dual Pl\"ucker coordinates with dual  Stiefel coordinates. The resulting polynomial, which we denote by $\widetilde{\mathcal{R}}_\mathcal{V}^\vee\in \F[Y_{i,j}: i\in [n-k], j\in [n+1]]$, is called the \emph{Chow form of $\mathcal{V}$ in dual Stiefel coordinates}.

We note that the Chow form $\widetilde{\mathcal{R}}_\mathcal{V}$ in primal Stiefel coordinates is a homogeneous polynomial of degree $(k+1)d$ in $(k+1)(n+1)$ variables, whereas the Chow form $\widetilde{\mathcal{R}}_\mathcal{V}^\vee$ in dual   Stiefel coordinates  is a homogeneous polynomial of degree $(n-k)d$ in $(n-k)(n+1)$ variables.
This suggests that it is more convenient to use the Chow form in primal (resp. dual) Stiefel coordinates when $k$ is small (resp. $n-k$ is small).\footnote{
While  both  $\mathcal{R}_\mathcal{V}$ and $\mathcal{R}_\mathcal{V}^\vee$   may be viewed as elements of $(R/I)_d$, the two (injective) maps $\mathcal{R}_\mathcal{V}\mapsto \widetilde{\mathcal{R}}_\mathcal{V}$ and $\mathcal{R}_\mathcal{V}^\vee\mapsto \widetilde{\mathcal{R}}_\mathcal{V}^\vee$ come from  different linear embedding of $(R/I)_d$ in  vector spaces of polynomials. As a result, the representation of $\mathcal{V}$ by the polynomial  $\widetilde{\mathcal{R}}_\mathcal{V}$  and the representation by $\widetilde{\mathcal{R}}_\mathcal{V}^\vee$ are not equally succinct in general.
}
 
\subsection{Explicit Constructions of Variety Evasive Subspace Families}\label{sec_construction}

Let  $n,d\in\N^+$, $k\in\{0,1,\dots,n-1\}$, and $\epsilon\in (0,1)$.  In this subsection, we prove the Main Theorem (\cref{thm_main})  by constructing explicit projective or affine $k$-subspace families that are $(n, d, \epsilon)$-evasive.
%The problem is trivial when $k=n$, as we just need to choose the singleton $\{\proj^n\}$ or $\{\aff^n\}$. 
%So assume $k<n$.

We first prove \cref{thm_main} in the projective case, and then derive the affine case from it by viewing $\aff^n$ as an open subset of $\proj^n$.
For the projective case, we present two constructions. The first one is simple and only uses $\epsilon$-hitting sets for low degree polynomials (\cref{lem_lowdeg}). But the size of the resulting subspace family is polynomial only when both $d$ and $k$ (or $n-k$) are bounded.
Next, we give a more sophisticated construction, which yields subspace families of polynomial size as long as $d$ is bounded.

\subsubsection{Simple Construction}

We first present a simple construction of $(n, d, \epsilon)$-evasive $k$-subspace families on $\proj^n$. 

First assume $k+1\leq n-k$. In this case, construct a $k$-subspace family $\mathcal{H}$ on $\proj^n$ as follows:
\begin{enumerate}
\item Use \cref{lem_lowdeg} to compute an $\epsilon$-hitting set $T$ for the family of polynomials $f\in \F[Y_{i,j}: i\in [k+1], j\in [n+1]]$ of degree at most $(k+1)d$ such that $|T|=\poly\left(\binom{(k+1)(n+1+d)}{(k+1)d}, 1/\epsilon\right)$.  
Think of $T$ as a collection of $(k+1)\times (n+1)$ matrices over $\F$. 
\item Initialize $\mathcal{H}=\emptyset$. For each matrix $A\in T$, if  $A$ has  full row rank $k+1$, add  to $\mathcal{H}$  the  $k$-subspace $W\in \pgr(k, n)$ with the generating matrix $A$.
\end{enumerate}

Next, assume $k+1>n-k$. In this case, construct $\mathcal{H}$ in a similar way, but use parity check matrices instead of generating matrices.
Namely, compute an $\epsilon$-hitting set $T$ for the family of polynomials $f\in \F[Y_{i,j}: i\in [n-k], j\in [n+1]]$ of degree at most $(n-k)d$ such that $|T|=\poly\left(\binom{(n-k)(n+1+d)}{(n-k)d}, 1/\epsilon\right)$.
Think of $T$ as a collection of $(n-k)\times (n+1)$ matrices over $\F$. 
For each matrix $A\in T$,  add  to $\mathcal{H}$  the  $k$-subspace $W\in \pgr(k, n)$ with the parity check matrix $A$.

This construction does give an $(n, d, \epsilon)$-evasive $k$-subspace family, as stated by the following lemma.

\begin{lemma}\label{lem_simple}
The $k$-subspace family $\mathcal{H}$ constructed above is $(n, d, \epsilon)$-evasive  and has size polynomial in 
$\min\left\{ \binom{(k+1)(n+1+d)}{(k+1)d}, \binom{(n-k)(n+1+d)}{(n-k)d}\right\}$ 
and $1/\epsilon$.
Moreover, the total time complexity of computing the linear equations defining the $k$-subspaces in $\mathcal{H}$ is polynomial in $|\mathcal{H}|$ (and $\log p$, if $\mathrm{char}(\F)=p>0$).
\end{lemma}

\begin{proof}
We only show that $\mathcal{H}$ is  $(n, d, \epsilon)$-evasive  since the rest of the lemma is obvious from the construction.
%\footnote{To see $\mathcal{H}$ has the desired cardinality, we use the fact that $\binom{ab+ac}{ab}$ is monotone in $a$ for $a,b,c\in\N^+$, so that we just need to compare $k-1$ and $n-k$ instead of the two binomial coefficients.
%Proof of monotonicity: for $a\leq a'$, we have $\binom{ab+ac}{ab}\leq \binom{a'b+ac}{a'b}=$}
 Let $\mathcal{F}$ be the family of all equidimensional subvarieties of $\proj^n$ of dimension $n-k-1$ and degree at most $d$. By \cref{lem_equi}, it suffices to prove that $\mathcal{H}$ is $(\mathcal{F}, \epsilon)$-evasive.
Consider any $\mathcal{V}\in\mathcal{F}$. We want to show that $\mathcal{V}\cap W=\emptyset$ for all but at most $\epsilon$-fraction of $W\in\mathcal{H}$.

First assume $k+1\leq n-k$, or equivalently, $\binom{(k+1)(n+1+d)}{(k+1)d}\leq \binom{(n-k)(n+1+d)}{(n-k)d}$. The Chow form $\widetilde{\mathcal{R}}_\mathcal{V}$ of $\mathcal{V}$ in Stiefel coordinates is a nonzero homogeneous polynomial in $\F[Y_{i,j}: i\in [k+1], j\in [n+1]]$ of degree $(k+1)\deg(\mathcal{V})\leq (k+1)d$.
By the choice of $T$, for all but at most $\epsilon$-fraction of $A\in T$, we have $\widetilde{\mathcal{R}}_\mathcal{V}(A)\neq 0$, which implies $\mathcal{V}\cap W=\emptyset$, where $A$ is a generating matrix of $W$.

By construction, $\mathcal{H}$ is the collection of $k$-subspaces corresponding to the matrices $A\in T$ of full row rank. 
So we have ignored the matrices that do not have full row rank.
But this does not increase the fraction of ``bad'' $W\in \mathcal{H}$ since if $A$ does not have full row rank, then the maximal minors of $A$ are all zero, and $\widetilde{\mathcal{R}}_\mathcal{V}(A)$ must be zero.  
It follows that $\mathcal{V}\cap W=\emptyset$ for all but at most $\epsilon$-fraction of $W\in\mathcal{H}$, as desired.

Now assume $k+1>n-k$. The proof in this case is similar and we omit the details. The only difference is that we use  the Chow form $\widetilde{\mathcal{R}}_\mathcal{V}^\vee$ in dual Stiefel coordinates instead of $\widetilde{\mathcal{R}}_\mathcal{V}$. 
\end{proof}

 \subsubsection{Improved Construction}
 
Before presenting the improved construction, we first introduce some notions from algebraic geometry.
 
\paragraph{Projections.} Suppose $W$ is a $k$-subspace of $\proj^n$, and  $\ell_1, \dots, \ell_{n-k}\in\F[X_1,\dots,X_{n+1}]$ are $n-k$ homogeneous linear polynomials such that  $W=\mathcal{V}(\ell_1,\dots,\ell_{n-k})$. 
Then we have a map $\pi: \proj^n\setminus W\to \proj^{n-k-1}$ defined by
\[
\pi: \mathbf{x}\mapsto (\ell_1(\mathbf{x}), \dots, \ell_{n-k}(\mathbf{x}))
\]
which is well-defined since $\ell_1, \dots, \ell_{n-k}$ never simultaneously vanish on $\proj^n\setminus W$.
We say $\pi$ is a \emph{projection} from $ \proj^n\setminus W$ to $\proj^{n-k-1}$ and $W$ is its \emph{center}. 
Note that if we lift $\pi$ to the  linear map $\pi': \aff^{n+1}\to\aff^{n-k}$ sending $\mathbf{x}\in\aff^{n+1}$ to $(\ell_1(\mathbf{x}), \dots, \ell_{n-k}(\mathbf{x}))\in\aff^{n-k}$, then the center $W$ is simply $\proj\ker(\pi')$.

We need the following lemma, whose proof can be found in \cite{Sha13}.
\begin{lemma}[{\cite[Section~I.5.3, Theorem~7]{Sha13}}]\label{lem_finitepr}
Suppose $\pi:  \proj^n\setminus W\to \proj^m$ is a projection with center $W$ and $\mathcal{V}$ is a subvariety of $\proj^n$ disjoint from $W$. Then $\pi$ restricts to a finite morphism from $\mathcal{V}$ to $\proj^m$.
\end{lemma}

\paragraph{Nondegenerate varieties.}
For a subvariety $\mathcal{V}\subseteq \proj^n$, denote by $\spn(\mathcal{V})$ the smallest projective subspace that contains $\mathcal{V}$.
We say  $\mathcal{V}$ is \emph{nondegenerate} if it is not contained in a hyperplane of $ \proj^n$, or equivalently, $\spn(\mathcal{V})=\proj^n$. 

We need the following  fact from algebraic geometry (see, e.g., \cite[Proposition~0]{EH87} or \cite[Corollary~18.12]{Har13}).

\begin{lemma}\label{lem_span}
The codimension of a nondegenerate irreducible subvariety $\mathcal{V}$ of $\proj^n$ is at most $\deg (\mathcal{V})-1$.
\end{lemma}

\paragraph{A two-step construction.}

We now give an improved construction of  $(n,d,\epsilon)$-evasive  $k$-subspace families on $\proj^n$ as follows.
\begin{enumerate}
\item If $k\leq d-2$, just use the previous simple construction. So assume $k>d-2$. Let $k'=d-2<k$, $n'=k'+n-k<n$, and $\epsilon_0=\epsilon/(2d)$.
\item Use \cref{cor_condenser} to construct a collection $\mathcal{H}_1$ of  $(n'+1)\times (n+1)$ matrix over $\F$  such that 
$|\mathcal{H}_1|=\poly(n, d/\epsilon)$ and
for every $(n+1)\times (n'+1)$ matrix $M$ over $\F$ of rank $n'+1$, all but at most $\epsilon_0$-fraction of $B\in U$ satisfies $\mathrm{rank}(BM)=n'+1$.

We abuse the notation and view $\mathcal{H}_1$ as a collection of linear maps from $\aff^{n+1}$ to $\aff^{n'+1}$.
Then for any linear subspace $W\subseteq \aff^{n+1}$ of dimension $n'+1$, we have $\dim(\pi(W))=\dim(W)=n'+1$ for all but at most $\epsilon_0$-fraction of $\pi\in \mathcal{H}_1$.

\item Construct a collection $\mathcal{H}_2$ of linear maps from $\aff^{n'+1}$ to $\aff^{n-k}$ as follows.
First assume $d>1$. Use \cref{lem_simple} to construct an $(n', d, \epsilon_0)$-evasive $k'$-subspace family $\mathcal{H}_2'$ on $\proj^{n'}$ of size  polynomial in 
$\min\left\{ \binom{(k'+1)(n'+1+d)}{(k'+1)d}, \binom{(n-k)(n'+1+d)}{(n-k)d}\right\}$ 
and $1/\epsilon_0$.
For each $k'$-subspace $W\in \mathcal{H}_2$, compute a surjective linear map $\pi_W: \aff^{n'+1}=\aff^{n-k+k'+1}\to \aff^{n-k}$ such that $W=\proj \ker(\pi_W)$. Let $\mathcal{H}_2=\{\pi_W: W\in \mathcal{H}_2'\}$.

If $d=1$, just let $\mathcal{H}_2$ be the singleton consisting of the identity map on $\aff^{n'+1}=\aff^{n-k}$.

\item  Initialize $\mathcal{H}=\emptyset$. For each $(\pi_1, \pi_2)\in \mathcal{H}_1\times \mathcal{H}_2$, if   $\dim(\ker(\pi_2\circ\pi_1))=k+1$, add the $k$-subspace $\proj\ker(\pi_2\circ\pi_1)$ to $\mathcal{H}$.\footnote{In fact,  $\dim(\ker(\pi_2\circ\pi_1))=k+1$ always holds since $\pi_1$ and $\pi_2$ are surjective. The fact that $\pi_1\in \mathcal{H}_1$ is surjective can be seen from the construction of lossless rank condensers in \cref{cor_condenser}.}
\end{enumerate}
 
We use the construction above  to prove the Main Theorem (\cref{thm_main}) in the projective case. For convenience, we restate it in the following form.

\begin{theorem}[Main Theorem in the projective case]\label{thm_mainproj}
The $k$-subspace family $\mathcal{H}$ constructed above is $(n, d, \epsilon)$-evasive  and has size 
$\poly(N(k,d,n), n, 1/\epsilon)$.
Moreover, the total time complexity of computing the linear equations defining the $k$-subspaces in $\mathcal{H}$ is polynomial in $|\mathcal{H}|$ (and $\log p$, if $\mathrm{char}(\F)=p>0$).
\end{theorem}

\begin{proof} 
The theorem follows from \cref{lem_simple} if $k\leq d-2$. So assume $k>d-2$.
We only show that $\mathcal{H}$ is $(n,d,\epsilon)$-evasive since the rest of the theorem is obvious from the construction.

Let $\mathcal{F}$ be the family of all irreducible subvarieties of $\proj^n$ of dimension $n-k-1$ and degree at most $d$.
By \cref{lem_irred}, it suffices to prove that $\mathcal{H}$ is $(\mathcal{F}, 2\epsilon_0)$-evasive.
Consider any $\mathcal{V}\in\mathcal{F}$. We want to show that $\mathcal{V}\cap W=\emptyset$ for all but at most $(2\epsilon_0)$-fraction of $W\in\mathcal{H}$.

By definition, $\mathcal{V}$ is a nondegenerate irreducible subvariety of $\spn(\mathcal{V})$.
By \cref{lem_span}, the codimension of $\mathcal{V}$ in $\spn(\mathcal{V})$ is at most $d-1$.
Therefore,
\[
\dim(\spn(\mathcal{V}))\leq \dim(\mathcal{V})+d-1= (n-k-1)+(d-1)=n'.
\]
Let $\Lambda\subseteq \proj^n$ be an $n'$-subspace that contains $\spn(\mathcal{V})$. 
%Let $\widetilde{\Lambda}\subseteq\aff^{n+1}$ be the linear $(n'+1)$-subspace such that $\Lambda=\proj\widetilde{\Lambda}$.
By the choice of $\mathcal{H}_1$, all but at most $\epsilon_0$-fraction of $\pi_1\in\mathcal{H}_1$ satisfies $\proj\ker(\pi_1)\cap \Lambda=\emptyset$.
Fix $\pi_1$ that satisfies this condition.
Then $\proj\ker(\pi_1)$ is disjoint from $\mathcal{V}\subseteq\Lambda$.
By \cref{lem_finitepr}, $\pi_1$ induces a finite morphism $\bar{\pi}_1: \mathcal{V}\to \proj^{n'}$.

Let $\mathcal{V}'=\bar{\pi}_1(\mathcal{V})\subseteq\proj^{n'}$. Then $\mathcal{V}'$ is a projective subvariety of dimension $\dim(\mathcal{V})=n-k-1$ and degree at most $d$.\footnote{The degree bound follows from, e.g., an inductive application of \cite[Proposition~5.5]{Mum76}.}
By the choice of $\mathcal{H}_2$, all but at most $\epsilon_0$-fraction of $\pi_2\in\mathcal{H}_2$ satisfies $\proj\ker(\pi_2)\cap \mathcal{V}'=\emptyset$.
Fix $\pi_2$ that satisfies this condition.
(If $d=1$ and $\pi_2$ is the identity map, we regard $\proj\ker(\pi_2)$ as an empty set, in which case this condition is also satisfied.)
By \cref{lem_finitepr}, $\pi_2$ induces a finite morphism $\bar{\pi}_2: \mathcal{V}'\to \proj^{n-k-1}$.
So we have a finite morphism $\bar{\pi}_2\circ \bar{\pi}_1:\mathcal{V}\to \proj^{n-k-1}$. Note that  $\bar{\pi}_2\circ \bar{\pi}_1$ is defined by restricting a projection with center $\proj \ker(\pi_2\circ \pi_1)$ to $\mathcal{V}$.
As $\bar{\pi}_2\circ \bar{\pi}_1$ is well-defined on $\mathcal{V}$, its center $\proj \ker(\pi_2\circ \pi_1)$ is disjoint from $\mathcal{V}$ and this also forces $\dim  \ker(\pi_2\circ \pi_1)=k+1$ by \cref{lem_dimcap}.

By the above argument and the construction of $\mathcal{H}$, all but at most $2\epsilon_0$-fraction of the $k$-subspaces in $\mathcal{H}$ are disjoint from $\mathcal{V}$, as desired.  
\end{proof}

\subsubsection{Derandomization of Noether's Normalization Lemma}

We now prove \cref{thm_dnoetherproj} and \cref{thm_dnoetheraff}. For convenience, we restate the theorems below.

\dnoetherproj*

\begin{proof}
If $r=n$, just use the identity map on $\aff^{n+1}=\aff^{r+1}$. So assume $r<n$. 
Use \cref{thm_mainproj} to construct an $(n, d, \epsilon)$-evasive $k$-subspace family $\mathcal{H}$ on $\proj^n$ of  size 
$\poly(N(k,d,n), n, 1/\epsilon)$. 
For each $k$-subspace $W\in \mathcal{H}$, compute a surjective linear map $\pi_W: \aff^{n+1}\to\aff^{r+1}$ such that $W=\proj \ker(\pi_W)$. Let $\mathcal{L}=\{\pi_W: W\in \mathcal{H}\}$.
Then $\mathcal{L}$ is a desired collection of linear maps by \cref{lem_finitepr}.
\end{proof}

\dnoetheraff*

\begin{proof}
If $r=n$,  just use the identity map on $\aff^{n}=\aff^{r}$. 
If $r=0$, use the only map  $\aff^{n}\to\aff^{0}$. 
So assume $0<r<n$. 
Regard $\aff^n$ as an open subset of $\proj^n$ via $(x_1,\dots,x_n)\mapsto (x_1,\dots,x_n, 1)$.
Similarly, regard $\aff^r$ as an open subset of $\proj^r$ via $(x_1,\dots,x_r)\mapsto (x_1,\dots,x_r, 1)$.
Let $H_\infty$ be the hyperplane at infinity of $\proj^n$ defined by $X_{n+1}=0$.

Use \cref{thm_mainproj} to construct an $(n-1, d, \epsilon)$-evasive $k$-subspace family $\mathcal{H}$ on $H_\infty\cong\proj^{n-1}$ of  size 
$\poly(N(k,d,n-1), n, 1/\epsilon)$. 
For each $W\in \mathcal{H}$, choose $n-k=r+1$ homogeneous linear polynomials $\ell_1,\dots,\ell_{r+1}\in \F[X_1,\dots, X_{n+1}]$
such that $\ell_{r+1}=X_{n+1}$, $\ell_1,\dots,\ell_r\in\F[X_1,\dots,X_n]$, and $W=\mathcal{V}(\ell_1,\dots,\ell_{r+1})$. 
This is possible as $W\subseteq H_\infty=\mathcal{V}(X_{n+1})$.
These $r+1$ linear polynomials determine a projection $\pi_W:  \proj^n\setminus W\to \proj^{r}$, defined by
\[
\mathbf{x}=(x_1,\dots,x_{n+1})\mapsto (\ell_1(\mathbf{x}), \dots, \ell_{r+1}(\mathbf{x}))= (\ell_1(\mathbf{x}), \dots, \ell_{r}(\mathbf{x}), x_{n+1}).
\]
As $x_{n+1}=1$ for $\mathbf{x}\in\aff^n$, we have $\pi_W(\aff^n)\subseteq\aff^r$.
Restricting $\pi_W$ on $\aff^n$ yields a map $\pi_W|_{\aff^n}:\aff^n\to \aff^r$,
which is a linear map as $\ell_1,\dots,\ell_r$ are homogeneous linear polynomials in $\F[X_1,\dots,X_n]$.
Let $\mathcal{L}=\{\pi_W|_{\aff^n}: W\in\mathcal{H}\}$.
 
Let $\mathcal{V}$ be a subvariety of $\aff^n$ of dimension $r$ and degree at most $d$.
Its projective closure $\mathcal{V}_\mathrm{cl}$ has dimension $\dim(\mathcal{V})=r$ and degree $\deg(\mathcal{V})\leq d$.
By the definition of  $\mathcal{V}_\mathrm{cl}$, none of the irreducible components of $\mathcal{V}_\mathrm{cl}$ is fully contained in $H_\infty$.
So by \cref{lem_hyperplane}, the projective subvariety $\mathcal{V}_\mathrm{cl}\cap H_\infty$  has dimension $r-1$ and degree at most $d$.

By the choice of $\mathcal{H}$, all but at most $\epsilon$-fraction of $W\in\mathcal{H}$ are disjoint from $\mathcal{V}_\mathrm{cl}\cap H_\infty$ and hence from $\mathcal{V}_\mathrm{cl}$.
So we just need to prove that for every $W\in\mathcal{H}$ disjoint from $\mathcal{V}_\mathrm{cl}$ and the corresponding projection $\pi_W$,
the map $\pi_W|_{\mathcal{V}}:  \mathcal{V}\to \aff^{r}$ is a surjective finite morphism.
This follows from \cref{lem_finitepr} and the fact that  $\mathcal{V}=\mathcal{V}_{\mathrm{cl}}\cap (\pi_W)^{-1}(\aff^r)$. 
\end{proof}

\subsubsection{Proof of the Main Theorem in the Affine Case}

We now prove \cref{thm_main} in the affine case.
Recall that we may view $\aff^n$ as an open subset of $\proj^n$ via the map $(x_1,\dots,x_n)\mapsto (x_1,\dots,x_n, 1)$. In this way, $\proj^n$ becomes the disjoint union of $\aff^n$ and the hyperplane at infinity $H_\infty$ defined by $X_{n+1}=0$.
% One way of constructing an $(n,d,\epsilon)$-evasive  affine $k$-subspace family $\mathcal{H}$ on $\aff^n$
%is first constructing an $(n,d,\epsilon')$-evasive projective $k$-subspace family $\mathcal{H}'$ on $\proj^n$ for some $\epsilon'>0$ depending on $\epsilon$,
%and then choosing $\mathcal{H}=\{W\cap\aff^n: W\in \mathcal{H}', W\not\subseteq H_\infty\}$. 
%The condition $W\not\subseteq H_\infty$ is satisfied by most $W\in \mathcal{H}'$ by $(n,d,\epsilon')$-evasiveness, so it only slightly increases the fraction of bad $k$-subspaces. The  $(n,d,\epsilon)$-evasiveness of $\mathcal{H}$ then follows easily from the $(n,d,\epsilon')$-evasiveness of $\mathcal{H}'$.

We use the following lemma to  reduce the affine case  to the projective case.
 
\begin{lemma}\label{lem_infty}
Let  $n,d\in\N^+$, $k\in\{0,1,\dots,n-1\}$, and $\epsilon'\in (0,1/2)$. 
Suppose $\mathcal{H}$ is an $(n, d, \epsilon')$-evasive $k$-subspace family on $\proj^n$.
Then 
\[
\mathcal{H}'=\{W\cap \aff^n: W\in \mathcal{H}, W\not\subseteq H_\infty\}
\] 
is an $(n, d, \epsilon)$-evasive affine $k$-subspace family on $\aff^n$, where $\epsilon=\epsilon'/(1-\epsilon')\leq 2\epsilon'$. Moreover, 
\[
\mathcal{H}''=\{W\in \mathcal{H}: W\not \subseteq H_\infty\}=\{W_\mathrm{cl}: W\in \mathcal{H}'\}
\] 
is an $(n, d, \epsilon)$-evasive $k$-subspace family on $\proj^n$.
\end{lemma}

\begin{proof}
By  $(n, d, \epsilon')$-evasiveness of $\mathcal{H}$, at most $\epsilon'$-fraction of $W\in\mathcal{H}$ are fully contained in $H_\infty$. Throwing away those $k$-subspaces fully contained in $H_\infty$ increases the error parameter $\epsilon'$ by at most a factor of $1/(1-\epsilon')$.
Therefore, 
$\mathcal{H}''=\{W\in \mathcal{H}: W\not \subseteq H_\infty\}$
 is $(n, d, \epsilon)$-evasive.
We want to prove that  $\mathcal{H}'=\{W\cap \aff^n: W\in \mathcal{H}''\}$ is also $(n, d, \epsilon)$-evasive. 

Consider a  subvariety $\mathcal{V}\subseteq \aff^n$ of degree at most $d$. Let $\mathcal{V}_1,\dots,\mathcal{V}_s$ be the irreducible components of $\mathcal{V}$.
The projective closure $\mathcal{V}_\mathrm{cl}$ of $\mathcal{V}$ has the irreducible components $(\mathcal{V}_1)_\mathrm{cl}, \dots, (\mathcal{V}_s)_\mathrm{cl}$.
Consider a $k$-subspace $W\in\mathcal{H}''$ that evades $\mathcal{V}_\mathrm{cl}$. We just need to prove that $W\cap  \aff^n$ evades $\mathcal{V}$. This is true since for each $i\in [s]$,
\[
\dim((W\cap \aff^n)\cap \mathcal{V}_i)\leq \dim(W\cap (\mathcal{V}_i)_\mathrm{cl}) \leq  \dim(W)+\dim((\mathcal{V}_i)_\mathrm{cl})-n=\dim(W\cap \aff^n)+\dim(\mathcal{V}_i)-n
\]
where the second inequality holds since $W$ evades $\mathcal{V}_\mathrm{cl}$ and the last equality uses the fact $W\not\subseteq H_\infty$.
\end{proof}

The affine case of \cref{thm_main} now follows easily.

\begin{proof}[Proof of \cref{thm_main} in the affine case.]
If $k=n$, just choose $\mathcal{H}=\aff^n$. Now assume $k<n$. Construct an $(n, d, \epsilon/2)$-evasive $k$-subspace family  $\mathcal{H}$ on $\proj^n$ using \cref{thm_mainproj}.
Then 
\[
\mathcal{H}':=\{W\cap \aff^n: W\in \mathcal{H}, W\not\subseteq H_\infty\}
\] 
is an $(n, d, \epsilon)$-evasive affine $k$-subspace family on $\aff^n$ by \cref{lem_infty}.
The nonhomogeneous linear equations defining $W\cap \aff^n\in\mathcal{H}'$ can be easily computed from the homogeneous linear equations defining $W\in\mathcal{H}$ by letting $X_{n+1}=1$.
\end{proof}

The proof of \cref{thm_main} is now complete.

%\begin{namedrem}{Side remark on the definition of evading}
%It is well known that for   projective subvarieties $\mathcal{V}_1,\mathcal{V}_2\subseteq \proj^n$ satisfying $\dim(\mathcal{V}_1) + \dim(\mathcal{V}_2) \geq n$,  the expected dimension $\dim(\mathcal{V}_1) + \dim(\mathcal{V}_2) - n$ is also the
%\emph{minimum} possible dimension of $\mathcal{V}_1\cap  \mathcal{V}_2$ (see \cref{lem_dimcap}). 
%However, 
\paragraph{Strengthening \cref{thm_main} in the affine case.}
For projective subvarieties $\mathcal{V}_1,\mathcal{V}_2\subseteq \proj^n$ such that $\dim(\mathcal{V}_1) + \dim(\mathcal{V}_2)\geq n$,  the minimum possible dimension of $\mathcal{V}_1\cap \mathcal{V}_2$ is $\dim(\mathcal{V}_1) + \dim(\mathcal{V}_2)- n$, as stated by \cref{lem_dimcap}.
Nevertheless, for two affine subvarieties  $\mathcal{V}_1,\mathcal{V}_2\subseteq \aff^n$, it is possible that the intersection of $\mathcal{V}_1$ and $\mathcal{V}_2$ is empty even if its expected dimension $\dim(\mathcal{V}_1) + \dim(\mathcal{V}_2) - n$ is nonnegative.
For example, the intersection of two distinct and parallel affine hyperplanes $\mathcal{V}_1,\mathcal{V}_2\subseteq\aff^n$ is always empty even if $n\geq 2$.
The reason this happens is that, while the dimension of $(\mathcal{V}_1)_\mathrm{cl}\cap (\mathcal{V}_2)_\mathrm{cl}$ is $n-2$ (as expected), this intersection is fully contained in the hyperplane $H_\infty$, which is excluded from $\aff^n$.
 
%For example, the intersection of two affine hyperplanes in general position has dimension $n-2$ for $n\geq 2$ but  two parallel affine hyperplanes are disjoint. 
 One may strengthen the definition of evading (\cref{defi_evade}) by requiring the intersection of $\mathcal{V}_1$ with every irreducible component of $\mathcal{V}_2$ to have \emph{exactly} the expected dimension.
It is possible to construct explicit affine $k$-subspace families satisfying \cref{thm_main} even under this stronger definition of evading. We sketch the  ideas as follows but omit the details.

First construct an $(n-1, d, \epsilon')$-evasive  $(k-1)$-subspace family $\mathcal{H}'$ on $H_\infty\cong \proj^{n-1}$ for some sufficiently small $\epsilon'$ depending on $\epsilon$. Then extend each $W\in\mathcal{H}'$ to a collection of $k$-subspaces by picking $p\in \aff^n$ and taking the $k$-subspace $J(W, p)$, where the coordinates of $p$ are chosen from an $\epsilon'$-hitting set for polynomials of degree at most $d$ given by \cref{lem_simple}. Call the resulting $k$-subspace family $\mathcal{H}$. It is easy to prove that $\mathcal{H}$ is $(n, d, O(\epsilon'))$-evasive. 

Furthermore, the affine $k$-subspace family $\{W\cap \aff^n: W\in \mathcal{H}\}$ is $(n, d, \epsilon)$-evasive even under the stronger definition of evading. To see this, consider an affine subvariety $\mathcal{V}\subseteq\aff^n$ of degree at most $d$. 
For most $W\in \mathcal{H}$, we have:
\begin{itemize}
\item For each irreducible component $\mathcal{V}_i$ of $\mathcal{V}$, the dimension of $(\mathcal{V}_i)_\mathrm{cl}\cap W$ is as expected by $(n, d, O(\epsilon'))$-evasiveness of $\mathcal{H}$ and \cref{lem_dimcap}. Call this dimension $d_i$, which is $-\infty$ if $(\mathcal{V}_i)_\mathrm{cl}\cap W=\emptyset$.
\item Moreover, the dimension of  $((\mathcal{V}_i)_\mathrm{cl}\cap H_\infty)\cap (W\cap H_\infty)$ is at most $d_i-1$ by $(n-1, d, \epsilon')$-evasiveness of $\mathcal{H}'$.
\item Therefore,  $\mathcal{V}_i\cap (W\cap \aff^n)$ has the expected dimension $d_i$ for each irreducible component $\mathcal{V}_i$ of $\mathcal{V}$.
\end{itemize}

\section{Lower Bound}\label{sec_lb}

We prove \cref{thm_lb} in this section. The main tool is the notion of \emph{Chow varieties}, which parameterize projective subvarieties. More precisely, they parametrize a generalization of projective subvarieties, called \emph{(effective) algebraic cycles} on a projective space.

\paragraph{Algebraic cycles.} An \emph{algebraic $r$-cycle} (or simply \emph{$r$-cycle})  on $\proj^n$  is a formal linear combination $D=\sum c_i \mathcal{V}_i$ of finitely many irreducible subvarieties $\mathcal{V}_i\subseteq \proj^n$ of dimension $r$, where the coefficients $c_i$ are integers. 
The \emph{degree} of  $D$ is $\deg(D):=\sum c_i \deg(\mathcal{V}_i)$.
The \emph{support} of $D$ is $\mathrm{supp}(D):=\bigcup_{c_i\neq 0} \mathcal{V}_i$.
An $r$-cycle is \emph{effective} if all its coefficients are nonnegative.
Denote by $C(r,d,n)$ the set of all effective $r$-cycles of degree $d$ on $\proj^{n}$.

\paragraph{Chow varieties.}
Let $k\in\{0,1,\dots,n-1\}$ and $r=n-k-1$.
The definition of Chow forms naturally extends to effective $r$-cycles. Namely, for an effective $r$-cycle $D=\sum_{i=1}^r c_i \mathcal{V}_i$ of degree $d$ on $\proj^{n}$, define the Chow form of $D$
to be $\mathcal{R}_D:=\prod_{i=1}^r \mathcal{R}_{\mathcal{V}_i}^{c_i}$. 

Note that $\mathcal{R}_D$ is a vector in $(R/I)_d$ and is uniquely determined by $D$ up to scalars. Write $[\mathcal{R}_D]$ for the point in $\proj (R/I)_d$ represented by  $\mathcal{R}_D$. Then we have  map $\psi: C(r, d, n)\to \proj (R/I)_d$, given by
\[
\psi: D\mapsto [\mathcal{R}_D],
\]
called the \emph{Chow embedding} of $C(r, d, n)$. Indeed, it embeds $C(r, d, n)$ in $\proj (R/I)_d$ as a projective subvariety, as stated by the following theorem of Chow and van der Waerden \cite{CvdW37}.
\begin{theorem}[\cite{CvdW37}]
The map $\psi$ is injective and its image is Zariski-closed.
\end{theorem}
A proof can also be found in \cite[Chapter~4]{GKZ94}. We identify $C(r, d, n)$ with its image under $\psi$ and view it as a projective variety. This variety is called the \emph{Chow variety} of effective $r$-cycles of degree $d$ on $\proj^n$.

\begin{exmp} 
Let $V$ be the subspace of homogeneous polynomials in $\F[X_1,\dots,X_{n+1}]$ of degree $d$.
Then $C(n-1, d, n)$ is simply the projective space $\proj V$ (see \cref{exmp_hyp}).
\end{exmp}

\begin{exmp}\label{exmp_grvar}
$C(r, 1, n)$ is the Grassmannian $\gr(r+1, n+1)$ (or $\pgr(r, n)$) embedded in $\proj^{\binom{n+1}{r+1}-1}=\proj^{\binom{n+1}{k+1}-1}$ via $\phi^\vee$ (see \cref{exmp_gr}).
\end{exmp}

\paragraph{The dimension of Chow varieties.} 
When $d=1$, the Chow variety $C(r, d, n)$ is just the Grassmannian $\gr(r+1, n+1)$ (see \cref{exmp_grvar}) and its dimension is well known to be $(r+1)(n-r)$ \cite{Har13}. 
When $d>1$, the dimension of   $C(r, d, n)$  was determined by Azcue in his Ph.D. thesis \cite{Azc93} and independently by Lehmann \cite{Leh17}. We state their result as follows.

\begin{theorem}[{\cite{Azc93, Leh17}}]\label{thm_chowdim}
For $d>1$ and $0\leq r<n$, the dimension of $C(r,d, n)$ is
\[
\max\left\{ d(r+1)(n-r), \binom{d+r+1}{r+1}-1+(r+2)(n-r-1)   \right\}.
\]
\end{theorem}
This theorem  was previously proved by Eisenbud and Harris \cite{EH92} for the special case $r=1$. 

\begin{rem}
To prove \cref{thm_lb}, we only need a lower bound for the dimension of the Chow variety, which is much easier to prove than \cref{thm_chowdim}. Indeed, it is not difficult to see that $d(r+1)(n-r)$ is the dimension of the space of unions of $d$ $r$-subspaces of $\proj^n$, and $\binom{d+r+1}{r+1}-1+(r+2)(n-r-1)$ is the dimension of the space of degree-$d$ hypersurfaces in $(r+1)$-subspaces of $\proj^n$.
\end{rem}

\paragraph{Lower bound via dimension counting.} We now restate \cref{thm_lb} and prove it using a dimension counting argument.

\lb*
\begin{proof}
Consider an arbitrary $k$-subspace $W\in\mathcal{H}$.
We may think of each point in $\proj (R/I)_d$ as a homogeneous polynomial  of degree $d$ in Pl\"ucker coordinates  modulo scalars and the ideal $I$ of  Pl\"ucker relations. We know Pl\"ucker coordinates always satisfy the  Pl\"ucker relations. So it makes sense to talk about if a point in $\proj (R/I)_d$ vanishes at $\phi(W)$ or not, as it does not depend on the choice of the homogeneous polynomial representing this point. Note that the constraint of $p\in \proj (R/I)_d$ vanishing at $\phi(W)$ is a linear equation in the homogeneous coordinates of $p$. So the set of points in $\proj (R/I)_d$ vanishing at $\phi(W)$ is a hyperplane of $\proj (R/I)_d$, which we denote by $H_W$.

Let $r=n-k-1$.
Assume $|\mathcal{H}|\leq  \dim(C(r,d, n)) $. Then we have 
\[
C(r,d, n)\cap \bigcap_{W\in\mathcal{H}} H_W\neq\emptyset
\]
 since taking the intersection with a hyperplane reduces the dimension of a projective subvariety by at most one (\cref{lem_hyperplane} or \cref{lem_dimcap}). So there exists an effective $r$-cycle $D=\sum_{i=1}^s c_i \mathcal{V}_i\in C(r, d, n)$, where $c_1,\dots,c_s>0$, such that $\psi(D)=[\mathcal{R}_D]$ vanishes at $\phi(W)$ for all $W\in\mathcal{H}$. 
 
Let $\mathcal{V}=\mathrm{supp}(D)=\bigcup_{i=1}^s \mathcal{V}_i$. Note  $\mathcal{V}\in \mathcal{F}$ since $\deg(\mathcal{V})\leq\deg(D)=d$.
For all  $W\in\mathcal{H}$, we know $\mathcal{R}_D=\prod_{i=1}^s \mathcal{R}_{\mathcal{V}_i}^{c_i}$ vanishes at $\phi(W)$, or equivalently, $\mathcal{R}_\mathcal{V}=\prod_{i=1}^s \mathcal{R}_{\mathcal{V}_i}$ vanishes at $\phi(W)$.  This implies $\mathcal{V}\cap W\neq \emptyset$ for all $W\in\mathcal{H}$. As $\mathcal{V}\in\mathcal{F}$, this contradicts our assumption about $\mathcal{H}$. We conclude 
\[
 |\mathcal{H}| \geq  \dim(C(r,d, n)) +1.
\]
The dimension of $C(r,d, n)$ is $(r+1)(n-r)$ when $d=1$ and is given by \cref{thm_chowdim} when $d>1$. Plugging in $r=n-k-1$ proves the theorem.
\end{proof}

\paragraph{A non-explicit construction.} Next, we show that the lower bound in \cref{thm_lb} is tight by matching it with a non-explicit construction. 

First, we need a bound for the degree of $C(r,d,n)$.
Define $M(r,d,n)$ by
\[
M(r,d,n):=\begin{cases} ((r+1)(n-r))!\prod_{i=1}^{r+1}\frac{(i-1)!}{(n-r+i-1)!} & \text{if } d=1, \\
 3^\lambda & \text{if } d>1.
\end{cases}
\]
where $\lambda:=\min\left\{\binom{n+d}{d}^{r+1}, \binom{n+d}{d}^{n-r}, \binom{\binom{n+1}{r+1}+d-1}{d}\right\}-1$.

\begin{lemma}\label{lem_degreebd}
The degree of $C(r,d,n)$ in $\proj(R/I)_d$ is at most $M(r,d,n)$.
\end{lemma}

\begin{proof}
When $d=1$, $C(r,d,n)$ is the Grassmannian $\pgr(r,n)$ and its degree under the Pl\"ucker embedding is known to be exactly $M(r,d,n)$ \cite{Kle76}. 

Now assume $d>1$. In this case, we use the following argument in \cite{Cat92}. 
Green and Morrison \cite{GM86} proved that the Chow variety $C(r,d,n)$ is defined by equations of degree at most three.
It follows from \cref{lem_hyperplane} (or B\'ezout's inequality \cite{Hei83}) that the degree of  $C(r,d,n)$ in $\proj(R/I)_d$ is bounded by $3^{\dim(\proj(R/I)_d)}$.
So it remains to prove that $\dim(\proj(R/I)_d)\leq \lambda$.

Recall that $R=\F\left[X_S: S\in \binom{[n+1]}{k+1}\right]$ where $k=n-r-1$.
So $\dim(R_d)=\binom{\binom{n+1}{k+1}+d-1}{d}=\binom{\binom{n+1}{r+1}+d-1}{d}$. Therefore, 
\[
\dim(\proj(R/I)_d)\leq \dim(\proj R_d)=\binom{\binom{n+1}{r+1}+d-1}{d}-1.
\]
On the other hand, the linear map $\mathcal{R}_\mathcal{V}\mapsto \widetilde{R}_\mathcal{V}$ embeds $\proj(R/I)_d$ in $\proj V$, where $V\subseteq \F[Y_{i,j}: i\in [k+1], j\in [n+1]]$ is the linear space of multihomogeneous polynomials of degree $(d,\dots,d)$ in the $k+1$ groups of variables $\{Y_{i_1},\dots,Y_{i,n+1}\}$, $i=1,\dots,k+1$. 
Therefore, 
\[
\dim(\proj(R/I)_d) \leq \dim(\proj V) =\binom{n+d}{d}^{k+1}-1=\binom{n+d}{d}^{n-r}-1.
\] 
Similarly, using the  linear map $\mathcal{R}_\mathcal{V}^\vee\mapsto \widetilde{R}_\mathcal{V}^\vee$, we get $\dim(\proj(R/I)_d)  \leq \binom{n+d}{d}^{r+1}-1$.
\end{proof}

The following theorem gives a non-explicit construction of  $\mathcal{H}$ whose cardinality matches the lower bound $\dim(C(r,d,n))+1$ in \cref{thm_lb}.

\begin{theorem}\label{thm_nonexplicit}
Let $n,d\in\N^+$, $k\in\{0,1,\dots,n-1\}$, $r=n-k-1$,  $t=\dim(C(r,d,n))$, and $\delta>0$. 
Let $\mathcal{F}$ be the family of equidimensional projective subvarieties of $\proj^n$ of dimension $n-k-1$ and degree at most $d$.
Let $S$ be a finite subset of $\F$ such that  
\[
|S|\geq M(r,d,n)\cdot t(k+1)d/\delta.
\]
Let $\mathcal{H}=\{W_1,\dots,W_{t+1}\}\subseteq\pgr(k,n)$ where the entries of the generating matrices of $W_1,\dots,W_{t+1}$ are chosen independently at random from $S$. Then with probability at least $1-\delta$, $\mathcal{H}$ is an $\mathcal{F}$-evasive $k$-subspace family on $\proj^n$. 
\end{theorem}

\begin{proof}
Whenever $\mathcal{H}$ fails to be $\mathcal{F}$-evasive, there exists an effective $r$-cycle $D$ of degree at most $d$ such that $\mathrm{supp}(D)$ intersects $W$ for all $W\in\mathcal{H}$. By adding extra $r$-subspaces to $D$ if necessary, we may assume the degree of $D$ is exactly $d$, i.e., $D\in C(r,d,n)$.

As argued in the proof of  \cref{thm_lb}, for a $k$-subspace $W\in\pgr(k,n)$, the support of $D\in C(r,d,n)$ intersects $W$ iff $D$ lies in a hyperplane $H_W$ of $\proj (R/I)_d$ corresponding to $W$.
So we just need to prove that the condition
\[
C(r,d, n)\cap \bigcap_{W\in\mathcal{H}} H_W=\emptyset
\] 
holds with probability at least $1-\delta$. Suppose $W_1,\dots,W_{i-1}$ are already chosen. Let $C=C(r,d, n)\cap \bigcap_{j=1}^{i-1} H_{W_j}$. By induction, it suffices to show that $\dim(C\cap H_{W_i})\leq \dim(C)-1$ holds with probability at least $1-\delta/t$. (Again, the dimension of an empty set is assumed to be $-\infty$.)

Consider an irreducible component $C_0$ of $C$. Fix $D=\sum_{j=1}^s c_j \mathcal{V}_j\in C_0$. The Chow form $\widetilde{R}_D=\prod_{j=1}^s \widetilde{R}_{\mathcal{V}_j}^{c_j}$ is a nonzero polynomial of degree $(k+1)d$. 
Let $M$ be the randomly chosen generating matrix of $W_i$.
By the Schwartz--Zippel Lemma \cite{Sch80, Zip79}, $\widetilde{R}_D(M)\neq 0$ holds with probability at least $1-(k+1)d/|S|$. When this occurs, we have $D\not \in H_{W_i}$ and hence  $\dim(C_0\cap H_{W_i})\leq \dim(C_0)-1$ by \cref{lem_hyperplane}.
The number of irreducible components of $C=C(r,d, n)\cap \bigcap_{j=1}^{i-1} H_{W_j}$ is bounded by $\deg(C)\leq \deg (C(r,d, n))\leq M(r,d,n)$, where the first inequality uses \cref{lem_hyperplane} and the second inequality holds by \cref{lem_degreebd}.
By the union bound,  the probability that $\dim(C\cap H_{W_i})\leq \dim(C)-1$ does not occur is bounded by
\[
M(r,d,n)\cdot (k+1)d/|S| \leq \delta/t
\]
as desired.
\end{proof}

\begin{rem}
While the cardinality $|\mathcal{H}|$ in \cref{thm_nonexplicit} is optimal, an unsatisfying issue here is that the elements in $S$ are huge when $d>1$. In particular, when $\min\{k, n-k\}$ is linear in $n$, these elements have exponential bit-length even if $d>1$ is bounded. This is due to the poor bound $M(r,d,n)$ for the number of irreducible components that we use. We suspect that this bound can be greatly improved.\footnote{It suffices to bound the number of the irreducible components of  $C(r,d, n)\cap \bigcap_{j=1}^{i-1} H_{W_j}$ whose general member is a subvariety (or even an irreducible subvariety) of $\proj^n$.}
In \cite[Exercise~3.28]{Kol13}, Koll{\'a}r outlined a method of proving a more effective bound for the number of irreducible components of $C(r,d,n)$. 
 Guerra \cite{Gue99} extended this method for Chow varieties associated with general projective varieties. Unfortunately, it is not clear to us if this method can be extended to bound the number of irreducible components of the intersection $C(r,d, n)\cap \bigcap_{j=1}^{i-1} H_{W_j}$. So we leave it as an open problem to obtain a more effective bound for the entries of the generating matrices in \cref{thm_nonexplicit}.
\end{rem}

\section{Application to PIT for Depth-4 Circuits}\label{sec_app}

In this section, we use explicit variety evasive subspace families to obtain a black-box PIT algorithm for non-SG $\Sigma\Pi\Sigma\Pi(k, r)$ circuits, thereby proving \cref{thm_pit}. The proof only uses the simple construction of variety evasive subspace families (\cref{lem_simple}).

We first define $\Sigma\Pi\Sigma\Pi(k, r)$ circuits and non-SG $\Sigma\Pi\Sigma\Pi(k, r)$ circuits.

\begin{definition}[$\Sigma\Pi\Sigma\Pi(k, r)$ circuit]
An algebraic circuit $C$ over $\F$ is a \emph{$\Sigma\Pi\Sigma\Pi(k, r)$ circuit} if
it has the form
\begin{equation}\label{eq_depth4}
C(X_1,\dots,X_n)=\sum_{i=1}^{k'} F_i=\sum_{i=1}^{k'} \prod_{j=1}^{d_i} Q_{i,j}
\end{equation}
where $k'\leq k$, $d_1,\dots, d_{k'}\in \N^+$, $F_i=\prod_{j=1}^{d_i} Q_{i,j}$ for $i\in [k']$, and each $Q_{i,j}$ is a polynomial in $X_1,\dots,X_n$ of degree at most $r$ over $\F$. 
The \emph{degree} of the circuit $C$ is defined to be $\max\{\deg(F_i): i\in [k']\}$.
In addition:
\begin{itemize}
\item $C$ is \emph{minimal} if $\sum_{i\in I} F_i\neq 0$ for all nonempty proper subset $I\subseteq [k']$.
\item $C$ is \emph{homogeneous} if all the polynomials $F_i$ are homogeneous of the same degree.
\item Let $\gcd(C):=\gcd(F_1,\dots,F_{k'})$. We say $C$ is \emph{simple} if $\gcd(C)=1$.
In general, we have $C=\gcd(C)\cdot\simple(C)$ where $\simple(C)$ is a simple $\Sigma\Pi\Sigma\Pi(k, r)$ circuit, called the \emph{simple part} of $C$. Note the simple part of a minimal $\Sigma\Pi\Sigma\Pi(k, r)$ circuit is still minimal.
\end{itemize}
The polynomial computed by $C$ is again denoted by $C$ by an abuse of notation.
\end{definition}

\begin{definition}[Non-SG circuit]
We say a minimal, simple, and homogeneous $\Sigma\Pi\Sigma\Pi(k, r)$ circuit  $C(X_1,\dots,X_n)=\sum_{i=1}^{k'} F_i$ as in \eqref{eq_depth4} is \emph{non-SG} if
there exists $i\in [k']$ such that
\[
\bigcap_{j\in [k']\setminus i}\mathcal{V}(F_j)\not\subseteq \mathcal{V}(F_i)
\]
where $\mathcal{V}(F)$ denotes the subvariety of $\proj^n$ defined by $F$.
More generally, a minimal and simple  $\Sigma\Pi\Sigma\Pi(k, r)$ circuit $C(X_1,\dots,X_n)=\sum_{i=1}^{k'} F_i$ of degree $d$ is non-SG if its homogenization 
\[
\widetilde{C}(X_1,\dots,X_{n+1})=\sum_{i=1}^{k'}  F_i(X_1/X_{n+1}, \dots, X_n/X_{n+1})\cdot X_{n+1}^{d}=\sum_{i=1}^{k'} \prod_{j=1}^{d_i'} \widetilde{Q}_{i,j}
\]
is non-SG, where $d_i'=d_i+(d-\deg(F_i))$ and $\widetilde{Q}_{i,j}$ equals the homogenization of $Q_{i,j}$ if $j\leq d_i$ and equals $X_{n+1}$ if $j>d_i$.
A minimal  $\Sigma\Pi\Sigma\Pi(k, r)$  circuit $C$ is non-SG if $\simple(C)$ is non-SG.
Finally, a $\Sigma\Pi\Sigma\Pi(k, r)$  circuit is non-SG if it has an equivalent minimal non-SG $\Sigma\Pi\Sigma\Pi(k, r)$  circuit.
\end{definition}

 We restate our result (\cref{thm_pit}) and then give a proof.

\pit*

\begin{proof}
If $n\leq k-1$, we may simply use \cref{lem_lowdeg} to construct a $\frac{1}{2}$-hitting set of size polynomial in $ \binom{n+d}{n}\leq \binom{k-1+d}{k-1}$ 
for $n$-variate polynomials of degree at most $d$, and then run the corresponding black-box PIT algorithm.
So assume $n>k-1$.

Consider a nonzero non-SG $\Sigma\Pi\Sigma\Pi(k, r)$ circuit $C$ of degree at most $d$.
We want to design a black-box PIT algorithm for $C$.
By replacing $C$ with an equivalent minimal non-SG circuit, we may assume $C$ is minimal.
Let $D=\gcd(C)$ and $E=\simple(C)$.
Let $\widetilde{C}$, $\widetilde{D}$, and $\widetilde{E}$ be the homogenization of $C$, $D$, and $E$ respectively.
Then $\widetilde{D}=\gcd(\widetilde{C})$, $\widetilde{E}=\simple(\widetilde{C})$, and $\widetilde{C}=\widetilde{D}\cdot\widetilde{E}$.

Let $\mathcal{H}$ be an affine $(k-1)$-subspace family on $\aff^n$ of size $\poly(\binom{k(n+1+r^k)}{kr^k}, d)$
such that $\mathcal{H}':=\{W_\mathrm{cl}: W\in \mathcal{H}\}$ is an $(n, r^k, \frac{1}{4d})$-evasive $(k-1)$-subspace family on $\proj^n$.
Such a family $\mathcal{H}$ can be computed using \cref{lem_infty} and  \cref{lem_simple}. 
We claim
\begin{enumerate}
\item $\widetilde{D}|_W\neq 0$ for all but at most $\frac{1}{4}$-fraction of $W\in \mathcal{H}'$, and
\item $\widetilde{E}|_W\neq 0$ for all but at most $\frac{1}{4}$-fraction of $W\in \mathcal{H}'$.
\end{enumerate}
Assume these two claims hold. Then for at least half of $W\in \mathcal{H}$, we have $\widetilde{C}|_{W_\mathrm{cl}}\neq 0$ and hence $C|_{W}=\widetilde{C}|_{W_\mathrm{cl}\cap \aff^n}\neq 0$, where we use the facts that $\widetilde{C}(X_1,\dots,X_n,1)$ equals $C(X_1,\dots,X_n)$ and  $W_\mathrm{cl}\cap \aff^n$ is dense in $W_\mathrm{cl}$.
The restriction of $C$ to each $W\cong\aff^{k-1}$ is a $(k-1)$-variate polynomial of degree at most $d$.
So to test if $C|_W$ is zero, we just need to use \cref{lem_lowdeg} to construct a hitting set in $W$ of size $\poly(\binom{k-1+d}{k-1})$ for $(k-1)$-variate polynomials of degree at most $d$. Take the union of these hitting sets to obtain a hitting set of size $\poly(\binom{k(n+1+r^k)}{kr^k}, d, \binom{k-1+d}{k-1})$ and we are done.

So it remains to prove the two claims. Note $\widetilde{D}$ is the product of at most $d$ factors whose degrees are bounded by $r$. The first claim then follows from the $(n, r^k, \frac{1}{4d})$-evasiveness of $\mathcal{H}'$ and the union bound.

Now we prove the second claim. By definition, $\widetilde{E}$ is a non-SG $\Sigma\Pi\Sigma\Pi(k, r)$ circuit.
Suppose it has the form
\begin{equation}\label{eq_simple}
\widetilde{E}=\sum_{i=1}^{k'} F_i=\sum_{i=1}^{k'} \prod_{j=1}^{d_i} Q_{i,j}
\end{equation}
where each $Q_{i,j}$ is a homogeneous polynomial of degree at most $r$.
As $\widetilde{E}$ is non-SG, there exists $i_0\in [k']$ such that
\[
\bigcap_{i\in [k']\setminus i_0}\mathcal{V}(F_i)\not\subseteq \mathcal{V}(F_{i_0})
\]
Without loss of generality, we may assume $i_0=k'$.
Note $\mathcal{V}(F_i)=\bigcup_{j=1}^{d_i} \mathcal{V}(Q_{i,j})$ for $i\in [k']$.
So there exists $(j_1,\dots, j_{k'-1})\in [d_1]\times\dots\times[d_{k'-1}]$ such that
\[
\bigcap_{i=1}^{k'-1}\mathcal{V}(Q_{i,j_i})\not\subseteq \mathcal{V}(F_{k'}). 
\]
Let $\mathcal{V}_0$ be an irreducible component of $\bigcap_{i=1}^{k'-1}\mathcal{V}(Q_{i,j_i})$ such that $\mathcal{V}_0\not\subseteq \mathcal{V}(F_{k'})$. Let $d_0=\dim(\mathcal{V}_0)\geq 0$.
By \cref{lem_hyperplane}, we have $d_0\geq n-k'+1$ and the variety $\mathcal{V}_0\cap \mathcal{V}(F_{k'})=\bigcup_{j=1}^{d_{k'}} (\mathcal{V}_0\cap \mathcal{V}(Q_{k', j}))$ has dimension at most $d_0-1$.
For each $j\in [d_{k'}]$, the degree of $\mathcal{V}_0\cap \mathcal{V}(Q_{k', j})$ is at most $r^k$ by \cref{lem_hyperplane} (or by B\'ezout's inequality \cite{Hei83}).
By $(n, r^k, \frac{1}{4d})$-evasiveness of $\mathcal{H}'$ and the union bound,
all but at most $\frac{1}{4}$-fraction of $W\in\mathcal{H}'$ evade $\mathcal{V}_0\cap \mathcal{V}(Q_{k', j})$ for $j=1,2,\dots,d_{k'}$.

Consider any $W\in\mathcal{H}'$ that evades $\mathcal{V}_0\cap \mathcal{V}(Q_{k', j})$ for $j=1,2,\dots,d_{k'}$. We just need to prove $\widetilde{E}|_W\neq 0$, or equivalently, $W\not\subseteq\mathcal{V}(\widetilde{E})$. Assume to the contrary that $W\subseteq\mathcal{V}(\widetilde{E})$.
Then $W\cap \mathcal{V}_0\subseteq\mathcal{V}(\widetilde{E})$. 
So
\begin{equation}\label{eq_compare}
W\cap \mathcal{V}_0=W\cap \mathcal{V}_0\cap \mathcal{V}(\widetilde{E})=W\cap \mathcal{V}_0\cap \mathcal{V}\left(\prod_{j=1}^{d_{k'}} Q_{k',j}\right)=\bigcup_{j=1}^{d_{k'}} (W\cap \mathcal{V}_0\cap\mathcal{V}(Q_{k', j}))
\end{equation}
where the second equality holds since $\widetilde{E}\equiv \prod_{j=1}^{d_{k'}} Q_{k',j}$ modulo the ideal $I_0:=\langle Q_{1,j_1},\dots, Q_{k'-1,j_{k'-1}}\rangle$ by \eqref{eq_simple} and $\mathcal{V}_0\subseteq \bigcap_{i=1}^{k'-1}\mathcal{V}(Q_{i,j_i})=\mathcal{V}(I_0)$.
We know the dimension of $\bigcup_{j=1}^{d_{k'}} (\mathcal{V}_0\cap \mathcal{V}(Q_{k', j}))$ is at most $d_0-1$.
So by the choice of $W$, the dimension of $\bigcup_{j=1}^{d_{k'}} (W\cap \mathcal{V}_0\cap\mathcal{V}(Q_{k', j}))$ is at most $(k-1)+(d_0-1)-n$.
However, by \cref{lem_dimcap}, the dimension of $W\cap \mathcal{V}_0$ is at least $(k-1)+d_0-n\geq 0$, where we use the fact $d_0\geq n-k'+1\geq n-k+1$.
This contradicts \eqref{eq_compare}.  So $\widetilde{E}|_W\neq 0$.
\end{proof}

\section{Open Problems and Future Directions} \label{sec_open}

We have seen that constructing explicit variety evasive subspace families is a natural problem that generalizes important problems in algebraic pseudorandomness and algebraic complexity theory, including deterministic black-box polynomial identity testing (evading varieties of codimension one) and constructing explicit lossless rank condensers (evading varieties of degree one). It is closely connected with advanced topics in algebraic geometry such as Chow forms and Chow varieties, and has applications to derandomizing PIT and non-explicit results in algebraic geometry like Noether's normalization lemma. 

There are many interesting open problems and potential future directions. We list some of them here.

\begin{enumerate}
\item \cref{thm_main}  focuses on subvarieties of bounded degree in a projective or affine space. Are there other interesting families of varieties for which we could construct explicit variety evasive subspace families? Families that are defined computation-theoretically may be particularly interesting, as many results of this kind are already known for polynomial identity testing.

\item Can explicit variety evasive subspace families be used to derandomize other non-explicit results in algebraic geometry?

\item Can our explicit construction in \cref{thm_main}  be improved? In the case $k=0$ and the case $d=1$, there are optimal or essentially optimal constructions, and our construction indeed degenerates into these constructions.
In general, however, there is a significant gap between the upper bound in \cref{thm_main} and the lower bound in \cref{thm_lb}.
\item Extending the notion of \emph{strong} lossless rank condensers \cite{FG15}, one could strengthen the definition of $(\mathcal{F},\epsilon)$-evasive subspace families in \cref{defi_evasive} by bounding the total deviation of the dimension instead of the number of bad subspaces.
At the same time, one could consider the setting where there is a gap between $\dim(\mathcal{V}_1)$ and $\codim(\mathcal{V}_2)$, as in typical applications of \emph{subspace designs} \cite{GX13, GR20, GRX21}.
 Alternatively, one could relax the definition by allowing $\dim(\mathcal{V}_1\cap\mathcal{V}_2)$ to be slightly greater than $\dim(\mathcal{V}_1) + \dim(\mathcal{V}_2) - n$, which is related to the notion of \emph{lossy} rank condensers in \cite{FG15}.
It is natural to study explicit constructions of these variants and their applications,  which can be seen as extensions of the theory of ``linear-algebraic pseudorandomness'' \cite{FG15} to a nonlinear setting.

\item Could our lower bound (\cref{thm_lb}) be extended to the affine case or to a ``lossy'' relaxation of the problem?

\item Is there a more effective bound for the entries of the generating matrices that are used in the non-explicit construction (\cref{thm_nonexplicit})?

\item When $n-k=O(1)$, our lower bound (\cref{thm_lb}) is only polynomial in $n$ and $d$. So one question is if there are explicit constructions of polynomial size when $n-k=O(1)$. 

As a concrete special case, consider the problem of constructing an explicit affine $(n-2)$-subspace family $\mathcal{H}$ on $\aff^n$ such that $\mathcal{H}$ is evasive for  degree-$d$ curves that are images of morphisms $\aff^1\to\aff^n$.
Note that for  $\varphi: \aff^1\to\aff^n$ corresponding to a ring homomorphism $\varphi^\sharp: \F[X_1,\dots,X_n]\to\F[Y]$,
an affine $(n-2)$-subspace defined by affine linear polynomials $\ell_1$ and $\ell_2$ evades the curve $\mathrm{Im}(\varphi)$ iff $\varphi^\sharp(\ell_1)$ and $\varphi^\sharp(\ell_2)$ have no common root. Using \emph{resultants}, we could reduce this problem to black-box PIT for symbolic determinants. 
Unconditionally, \cref{thm_main} also yields an explicit construction of polynomial size when $d=O(\log n)$.
We are not aware of any unconditional derandomization whose time complexity is subexponential in $\min\{n,d\}$, however.

%\item \label{enum_nonexplicit} Our lower bound (\cref{thm_lb}) for $|\mathcal{H}|$ is polynomial when either $d$ or $r=n-k-1$ is fixed. In this case, is there a non-explicit construction of $\mathcal{H}$ of optimal cardinality and \emph{polynomial bit-length}? 
%We believe and conjecture that this is true.
%
%It is easy to prove the existence of $\mathcal{H}$ of optimal cardinality. However, achieving polynomial bit-length  seems to require more work due to the complexity of the Chow variety $C(r, d, n)$. We leave this question for future investigation. 
%%For example, we only know doubly exponential upper bounds for the degree of $C(r,d, n)$ even when $d$ or $r$ is fixed \cite{Cat92}.
%One tool that may be useful is a doubly exponential upper bound for the number of irreducible components of $C(r,d,n)$, which becomes singly exponential when either $d$ or $r$ is fixed \cite[Exercise~3.28]{Kol13}.

\end{enumerate}

\section*{Acknowledgements}

\noindent We thank Nitin Saxena, Noga Ron-Zewi, Amit Sinhababu, and Suryajith Chillara for helpful discussions.
%We also thank the anonymous reviewers of CCC'21 for their careful reading of our manuscript and their insightful suggestions that helped improve this paper.

\bibliographystyle{alpha}
\bibliography{ref}

\newcommand{\etalchar}[1]{$^{#1}$}
\begin{thebibliography}{KRZSW23}

\bibitem[AGKS15]{AGKS15}
Manindra Agrawal, Rohit Gurjar, Arpita Korwar, and Nitin Saxena.
\newblock Hitting-sets for {ROABP} and sum of set-multilinear circuits.
\newblock {\em SIAM Journal on Computing}, 44(3):669--697, 2015.

\bibitem[ASSS16]{ASSS16}
Manindra Agrawal, Chandan Saha, Ramprasad Saptharishi, and Nitin Saxena.
\newblock Jacobian hits circuits: Hitting sets, lower bounds for depth-d
  occur-k formulas and depth-3 transcendence degree-k circuits.
\newblock {\em SIAM Journal on Computing}, 45(4):1533--1562, 2016.

\bibitem[AV08]{AV08}
Manindra Agrawal and V~Vinay.
\newblock Arithmetic circuits: A chasm at depth four.
\newblock In {\em Proceedings of the 49th Annual IEEE Symposium on Foundations
  of Computer Science}, pages 67--75. IEEE, 2008.

\bibitem[Azc92]{Azc93}
Pablo Azcue.
\newblock {\em On the dimension of the {Chow} varieties}.
\newblock PhD thesis, Harvard University, 1992.

\bibitem[BMS13]{BMS13}
Malte Beecken, Johannes Mittmann, and Nitin Saxena.
\newblock Algebraic independence and blackbox identity testing.
\newblock {\em Information and Computation}, 222:2--19, 2013.

\bibitem[Bog05]{Bog05}
Andrej Bogdanov.
\newblock Pseudorandom generators for low degree polynomials.
\newblock In {\em Proceedings of the 37th Annual ACM Symposium on Theory of
  Computing}, pages 21--30. ACM, 2005.

\bibitem[BP20]{BP20}
Markus Bl{\"a}ser and Anurag Pandey.
\newblock Polynomial identity testing for low degree polynomials with optimal
  randomness.
\newblock In {\em Approximation, Randomization, and Combinatorial Optimization.
  Algorithms and Techniques (APPROX/RANDOM)}, pages 8:1--8:13. Schloss Dagstuhl
  -- Leibniz-Zentrum f{\"u}r Informatik, 2020.

\bibitem[Bsh14]{Bsh14}
Nader Bshouty.
\newblock Testers and their applications.
\newblock In {\em Proceedings of the 5th Conference on Innovations in
  Theoretical Computer Science}, pages 327--352. ACM, 2014.

\bibitem[Cat92]{Cat92}
Fabrizio Catanese.
\newblock Chow varieties, {Hilbert} schemes, and moduli spaces of surfaces of
  general type.
\newblock {\em Journal of Algebraic Geometry}, 1(4):561--595, 1992.

\bibitem[Cay60]{Cay60}
Arthur Cayley.
\newblock On a new analytical representation of curves in space.
\newblock {\em The Quarterly Journal of Pure and Applied Mathematics},
  3:225--236, 1860.

\bibitem[CGS21]{CGS20}
Arkadev Chattopadhyay, Ankit Garg, and Suhail Sherif.
\newblock Towards stronger counterexamples to the log-approximate-rank
  conjecture.
\newblock In {\em 41st IARCS Annual Conference on Foundations of Software
  Technology and Theoretical Computer Science (FSTTCS 2021)}, pages
  13:1--13:16. Schloss Dagstuhl -- Leibniz-Zentrum f{\"u}r Informatik, 2021.

\bibitem[CTS13]{CT13}
Gil Cohen and Amnon Ta-Shma.
\newblock Pseudorandom generators for low degree polynomials from algebraic
  geometry codes.
\newblock In {\em Electronic Colloquium on Computational Complexity (ECCC)},
  volume~20, page 155, 2013.

\bibitem[CvdW37]{CvdW37}
Wei-Liang Chow and B.L. van~der Waerden.
\newblock Zur algebraischen {Geometrie}. {IX}.
\newblock {\em Mathematische Annalen}, 113(1):692--704, 1937.

\bibitem[DDS21]{DDS21}
Pranjal Dutta, Prateek Dwivedi, and Nitin Saxena.
\newblock Deterministic identity testing paradigms for bounded top-fanin
  depth-4 circuits.
\newblock In {\em 36th Computational Complexity Conference (CCC 2021)}, pages
  11:1--11:27. Schloss Dagstuhl -- Leibniz-Zentrum f{\"u}r Informatik, 2021.

\bibitem[DGW09]{DGW09}
Zeev Dvir, Ariel Gabizon, and Avi Wigderson.
\newblock Extractors and rank extractors for polynomial sources.
\newblock {\em Computational Complexity}, 18(1):1--58, 2009.

\bibitem[DKL14]{DKL14}
Zeev Dvir, J{\'a}nos Koll{\'a}r, and Shachar Lovett.
\newblock Variety evasive sets.
\newblock {\em Computational Complexity}, 23(4):509--529, 2014.

\bibitem[DL78]{DL78}
Richard~A. Demillo and Richard~J. Lipton.
\newblock A probabilistic remark on algebraic program testing.
\newblock {\em Information Processing Letters}, 7(4):193--195, 1978.

\bibitem[DL12]{DL12}
Zeev Dvir and Shachar Lovett.
\newblock Subspace evasive sets.
\newblock In {\em Proceedings of the 44th Annual ACM Symposium on Theory of
  Computing}, pages 351--358. ACM, 2012.

\bibitem[DS95]{DS95}
John Dalbec and Bernd Sturmfels.
\newblock Introduction to {Chow} forms.
\newblock {\em Invariant Methods in Discrete and Computational Geometry}, pages
  37--58, 1995.

\bibitem[DS07]{DS07}
Zeev Dvir and Amir Shpilka.
\newblock Locally decodable codes with two queries and polynomial identity
  testing for depth 3 circuits.
\newblock {\em SIAM Journal on Computing}, 36(5):1404--1434, 2007.

\bibitem[Dub93]{Dub93}
Thomas~W Dub{\'e}.
\newblock A combinatorial proof of the effective {Nullstellensatz}.
\newblock {\em Journal of Symbolic Computation}, 15(3):277--296, 1993.

\bibitem[Dvi12]{Dvi12}
Zeev Dvir.
\newblock Extractors for varieties.
\newblock {\em Computational Complexity}, 21(4):515--572, 2012.

\bibitem[EH87]{EH87}
David Eisenbud and Joe Harris.
\newblock On varieties of minimal degree.
\newblock In {\em Proceedings of Symposia in Pure Mathematics}, volume~46,
  pages 3--13. American Mathematical Society, 1987.

\bibitem[EH92]{EH92}
David Eisenbud and Joe Harris.
\newblock The dimension of the {Chow} variety of curves.
\newblock {\em Compositio Mathematica}, 83(3):291--310, 1992.

\bibitem[Eis95]{Eis13}
David Eisenbud.
\newblock {\em Commutative Algebra: with a View Toward Algebraic Geometry}.
\newblock Springer-Verlag, 1995.

\bibitem[FG15]{FG15}
Michael~A. Forbes and Venkatesan Guruswami.
\newblock Dimension expanders via rank condensers.
\newblock In {\em Approximation, Randomization, and Combinatorial Optimization.
  Algorithms and Techniques (APPROX/RANDOM)}. Schloss Dagstuhl --
  Leibniz-Zentrum f{\"u}r Informatik, 2015.

\bibitem[For14]{For14}
Michael~A. Forbes.
\newblock {\em Polynomial identity testing of read-once oblivious algebraic
  branching programs}.
\newblock PhD thesis, Massachusetts Institute of Technology, 2014.

\bibitem[FS12]{FS12}
Michael~A. Forbes and Amir Shpilka.
\newblock On identity testing of tensors, low-rank recovery and compressed
  sensing.
\newblock In {\em Proceedings of the 44th Annual ACM Symposium on Theory of
  Computing}, pages 163--172. ACM, 2012.

\bibitem[FS13]{FS13}
Michael~A. Forbes and Amir Shpilka.
\newblock Explicit {Noether} normalization for simultaneous conjugation via
  polynomial identity testing.
\newblock In {\em Approximation, Randomization, and Combinatorial Optimization.
  Algorithms and Techniques (APPROX/RANDOM)}, pages 527--542. Springer Berlin
  Heidelberg, 2013.

\bibitem[FS18]{FS17}
Michael~A. Forbes and Amir Shpilka.
\newblock A {PSPACE} construction of a hitting set for the closure of small
  algebraic circuits.
\newblock In {\em Proceedings of the 50th Annual ACM Symposium on Theory of
  Computing}, pages 1180--1192. ACM, 2018.

\bibitem[FSS14]{FSS14}
Michael~A. Forbes, Ramprasad Saptharishi, and Amir Shpilka.
\newblock Hitting sets for multilinear read-once algebraic branching programs,
  in any order.
\newblock In {\em Proceedings of the 46th Annual ACM Symposium on Theory of
  Computing}, pages 867--875. ACM, 2014.

\bibitem[Ful97]{Ful97}
William Fulton.
\newblock {\em Young Tableaux: With Applications to Representation Theory and
  Geometry}, volume~35.
\newblock Cambridge University Press, 1997.

\bibitem[GH93]{GH93}
Marc Giusti and Joos Heintz.
\newblock La d\'etermination des points isol\'es et de la dimension d'une
  vari\'et\'e alg\'ebrique peut se faire en temps polynomial.
\newblock {\em Computational Algebraic Geometry and Commutative Algebra
  (Cortona, 1991)}, pages 216--256, 1993.

\bibitem[GHL{\etalchar{+}}00]{GHLMS00}
Marc Giusti, Klemens H{\"a}gele, Gr{\'e}goire Lecerf, Jo{\"e}l Marchand, and
  Bruno Salvy.
\newblock The projective {Noether} {Maple} package: computing the dimension of
  a projective variety.
\newblock {\em Journal of Symbolic Computation}, 30(3):291--307, 2000.

\bibitem[GK16]{GK16}
Venkatesan Guruswami and Swastik Kopparty.
\newblock Explicit subspace designs.
\newblock {\em Combinatorica}, 36(2):161--185, 2016.

\bibitem[GKZ94]{GKZ94}
Israel~M. Gelfand, Mikhail~M. Kapranov, and Andrei~V. Zelevinsky.
\newblock {\em Discriminants, Resultants and Multidimensional Determinants}.
\newblock Birkh{\"a}user, 1994.

\bibitem[GM86]{GM86}
Mark~L. Green and Ian Morrison.
\newblock The equations defining {Chow} varieties.
\newblock {\em Duke Mathematical Journal}, 53(3):733--747, 1986.

\bibitem[GR08]{GR08}
Ariel Gabizon and Ran Raz.
\newblock Deterministic extractors for affine sources over large fields.
\newblock {\em Combinatorica}, 28(4):415--440, 2008.

\bibitem[GRX21]{GRX21}
Venkatesan Guruswami, Nicolas Resch, and Chaoping Xing.
\newblock Lossless dimension expanders via linearized polynomials and subspace
  designs.
\newblock {\em Combinatorica}, pages 1--35, 2021.

\bibitem[GRZ21]{GR20}
Zeyu Guo and Noga Ron-Zewi.
\newblock Efficient list-decoding with constant alphabet and list sizes.
\newblock {\em IEEE Transactions on Information Theory}, 68(3):1663--1682,
  2021.

\bibitem[GSS19]{GSS18}
Zeyu Guo, Nitin Saxena, and Amit Sinhababu.
\newblock Algebraic dependencies and {PSPACE} algorithms in approximative
  complexity over any field.
\newblock {\em Theory of Computing}, 15(16):1--30, 2019.

\bibitem[Gue99]{Gue99}
Lucio Guerra.
\newblock Complexity of {Chow} varieties and number of morphisms on surfaces of
  general type.
\newblock {\em Manuscripta Mathematica}, 98(1):1--8, 1999.

\bibitem[Guo21]{Guo21}
Zeyu Guo.
\newblock Variety evasive subspace families.
\newblock In {\em 36th Computational Complexity Conference (CCC 2021)}, pages
  20:1--20:33. Schloss Dagstuhl -- Leibniz-Zentrum f{\"u}r Informatik, 2021.

\bibitem[Gup14]{Gup14}
Ankit Gupta.
\newblock Algebraic geometric techniques for depth-4 {PIT} \&
  {Sylvester}-{Gallai} conjectures for varieties.
\newblock In {\em Electronic Colloquium on Computational Complexity (ECCC)},
  volume~21, page 130, 2014.

\bibitem[GWX16]{GWX16}
Venkatesan Guruswami, Carol Wang, and Chaoping Xing.
\newblock Explicit list-decodable rank-metric and subspace codes via subspace
  designs.
\newblock {\em IEEE Transactions on Information Theory}, 62(5):2707--2718,
  2016.

\bibitem[GX13]{GX13}
Venkatesan Guruswami and Chaoping Xing.
\newblock List decoding {Reed-Solomon}, {Algebraic-Geometric}, and {Gabidulin}
  subcodes up to the {Singleton} bound.
\newblock In {\em Proceedings of the 45th Annual ACM Symposium on Theory of
  Computing}, pages 843--852. ACM, 2013.

\bibitem[GXY18]{GXY18}
Venkatesan Guruswami, Chaoping Xing, and Chen Yuan.
\newblock Subspace designs based on algebraic function fields.
\newblock {\em Transactions of the American Mathematical Society},
  370(12):8757--8775, 2018.

\bibitem[Har92]{Har13}
Joe Harris.
\newblock {\em Algebraic Geometry: A First Course}.
\newblock Springer-Verlag, 1992.

\bibitem[Hei83]{Hei83}
Joos Heintz.
\newblock Definability and fast quantifier elimination in algebraically closed
  fields.
\newblock {\em Theoretical Computer Science}, 24(3):239--277, 1983.

\bibitem[Hil90]{Hil90}
David Hilbert.
\newblock Ueber die {Theorie} der algebraischen {Formen}.
\newblock {\em Mathematische Annalen}, 36(4):473--534, 1890.

\bibitem[HS80]{HS80}
Joos Heintz and Claus-Peter Schnorr.
\newblock Testing polynomials which are easy to compute (extended abstract).
\newblock In {\em Proceedings of the 12th annual ACM Symposium on Theory of
  Computing}, pages 262--272. ACM, 1980.

\bibitem[HS81]{HS81}
Joos Heintz and Malte Sieveking.
\newblock Absolute primality of polynomials is decidable in random polynomial
  time in the number of variables.
\newblock In {\em International Colloquium on Automata, Languages, and
  Programming}, pages 16--28. Springer Berlin Heidelberg, 1981.

\bibitem[JKSS04]{JKSS04}
Gabriela Jeronimo, Teresa Krick, Juan Sabia, and Mart{\'\i}n Sombra.
\newblock The computational complexity of the {Chow} form.
\newblock {\em Foundations of Computational Mathematics}, 4(1):41--117, 2004.

\bibitem[JT13]{JT13}
Michael Joswig and Thorsten Theobald.
\newblock {\em Polyhedral and Algebraic Methods in Computational Geometry}.
\newblock Springer Science \& Business Media, 2013.

\bibitem[Kle76]{Kle76}
Steven Kleiman.
\newblock Problem 15. rigorous foundation of {Schubert}’s enumerative
  calculus.
\newblock In {\em Proceedings of Symposia in Pure Mathematics}, volume~28,
  pages 445--482, 1976.

\bibitem[Kol88]{Kol88}
J{\'a}nos Koll{\'a}r.
\newblock Sharp effective {Nullstellensatz}.
\newblock {\em Journal of the American Mathematical Society}, pages 963--975,
  1988.

\bibitem[Kol96]{Kol13}
J{\'a}nos Koll{\'a}r.
\newblock {\em Rational Curves on Algebraic Varieties}.
\newblock Springer, 1996.

\bibitem[Kri02]{Kri02}
Teresa Krick.
\newblock Straight-line programs in polynomial equation solving.
\newblock {\em Foundations of Computational Mathematics}, 312:96--136, 2002.

\bibitem[KRZSW23]{KRSW18}
Swastik Kopparty, Noga Ron-Zewi, Shubhangi Saraf, and Mary Wootters.
\newblock Improved list decoding of folded {R}eed-{S}olomon and multiplicity
  codes.
\newblock {\em SIAM Journal on Computing}, 52(3):794--840, 2023.

\bibitem[KS01]{KS01}
Adam~R. Klivans and Daniel Spielman.
\newblock Randomness efficient identity testing of multivariate polynomials.
\newblock In {\em Proceedings of the 33rd Annual ACM Symposium on Theory of
  Computing}, pages 216--223. ACM, 2001.

\bibitem[KS11]{KS11}
Zohar~S. Karnin and Amir Shpilka.
\newblock Black box polynomial identity testing of generalized depth-3
  arithmetic circuits with bounded top fan-in.
\newblock {\em Combinatorica}, 31(3):333, 2011.

\bibitem[Lan12]{Lan12}
Joseph~M. Landsberg.
\newblock {\em Tensors: Geometry and Applications}.
\newblock Graduate studies in mathematics. American Mathematical Society, 2012.

\bibitem[Lan15]{Lan15}
Joseph~M. Landsberg.
\newblock Geometric complexity theory: an introduction for geometers.
\newblock {\em Annali dell'universita'di Ferrara}, 61(1):65--117, 2015.

\bibitem[Leh17]{Leh17}
Brian Lehmann.
\newblock Asymptotic behavior of the dimension of the {Chow} variety.
\newblock {\em Advances in Mathematics}, 308:815--835, 2017.

\bibitem[Len90]{Len90}
Hendrik~W. Lenstra, Jr.
\newblock Algorithms for finite fields.
\newblock {\em Number theory and cryptography}, 154:76--85, 1990.

\bibitem[LST24]{LST21}
Nutan Limaye, Srikanth Srinivasan, and S{\'e}bastien Tavenas.
\newblock Superpolynomial lower bounds against low-depth algebraic circuits.
\newblock {\em Communications of the ACM}, 67(2):101--108, 2024.

\bibitem[Lu12]{Lu12}
Chi-Jen Lu.
\newblock Hitting set generators for sparse polynomials over any finite fields.
\newblock In {\em Proceedings of the 27th Conference on Computational
  Complexity}, pages 280--286. IEEE, 2012.

\bibitem[Muk16]{Muk16}
Partha Mukhopadhyay.
\newblock Depth-4 identity testing and {Noether’s} normalization lemma.
\newblock In {\em Proceedings of the 11th International Computer Science
  Symposium in Russia}, pages 309--323. Springer International Publishing,
  2016.

\bibitem[Mul17]{Mul17}
Ketan Mulmuley.
\newblock Geometric complexity theory {V}: Efficient algorithms for {Noether}
  normalization.
\newblock {\em Journal of the American Mathematical Society}, 30(1):225--309,
  2017.

\bibitem[Mum76]{Mum76}
David Mumford.
\newblock {\em Algebraic Geometry I: Complex Projective Varieties}.
\newblock Springer-Verlag, 1976.

\bibitem[Nag62]{Nag62}
Masayoshi Nagata.
\newblock Local rings.
\newblock {\em Interscience Tracts in Pure and Applied Mathematics}, 1962.

\bibitem[Noe26]{Noe26}
Emmy Noether.
\newblock Der {E}ndlichkeitssatz der {I}nvarianten endlicher linearer {G}ruppen
  der {C}harakteristik p.
\newblock {\em Nachrichten von der Gesellschaft der Wissenschaften zu
  G{\"o}ttingen, Mathematisch-Physikalische Klasse}, 1926:28--35, 1926.

\bibitem[PS21]{PS20_2}
Shir Peleg and Amir Shpilka.
\newblock Polynomial time deterministic identity testing algorithm for
  {$\Sigma^{[3]}\Pi\Sigma\Pi^{[2]}$} circuits via {E}delstein--{K}elly type
  theorem for quadratic polynomials.
\newblock In {\em Proceedings of the 53rd Annual ACM Symposium on Theory of
  Computing}, pages 259--271. ACM, 2021.

\bibitem[PS22a]{PS22}
Luis~M. Pardo and Daniel Sebasti\'an.
\newblock A promenade through correct test sequences {I}: Degree of
  constructible sets, {B}\'ezout's inequality and density.
\newblock {\em Journal of Complexity}, 68:101588, 2022.

\bibitem[PS22b]{PS20}
Shir Peleg and Amir Shpilka.
\newblock A generalized {S}ylvester--{G}allai-type theorem for quadratic
  polynomials.
\newblock {\em Forum of Mathematics, Sigma}, 10:e112, 2022.

\bibitem[Sax09]{Sax09}
Nitin Saxena.
\newblock Progress on polynomial identity testing.
\newblock {\em Bulletin of the {EATCS}}, 99:49--79, 2009.

\bibitem[Sax13]{Sax13}
Nitin Saxena.
\newblock Progress on polynomial identity testing - {II}.
\newblock {\em Electronic Colloquium on Computational Complexity {(ECCC)}},
  20:186, 2013.

\bibitem[Sch80]{Sch80}
Jacob~T. Schwartz.
\newblock Fast probabilistic algorithms for verification of polynomial
  identities.
\newblock {\em Journal of the ACM}, 27(4):701--717, 1980.

\bibitem[Sha94]{Sha13}
Igor~R. Shafarevich.
\newblock {\em Basic Algebraic Geometry 1: Varieties in Projective Space}.
\newblock Springer-Verlag, 1994.

\bibitem[Shp20]{Shp19}
Amir Shpilka.
\newblock {S}ylvester-{G}allai type theorems for quadratic polynomials.
\newblock {\em Discrete Analysis}, 2020.

\bibitem[SS12]{SS12}
Nitin Saxena and Comandur Seshadhri.
\newblock Blackbox identity testing for bounded top-fanin depth-3 circuits: The
  field doesn't matter.
\newblock {\em SIAM Journal on Computing}, 41(5):1285--1298, 2012.

\bibitem[SY10]{SY10}
Amir Shpilka and Amir Yehudayoff.
\newblock {\em Arithmetic Circuits: A Survey of Recent Results and Open
  Questions}.
\newblock Now Publishers Inc, 2010.

\bibitem[Zip79]{Zip79}
Richard Zippel.
\newblock Probabilistic algorithms for sparse polynomials.
\newblock In {\em International Symposium on Symbolic and Algebraic
  Manipulation}, pages 216--226. Springer Berlin Heidelberg, 1979.

\end{thebibliography}

\end{document}